\documentclass[journal,twocolumn]{IEEEtran}

\usepackage{graphicx}
\usepackage [noadjust]{cite} 
\usepackage{bm}  
\usepackage{amsmath}
\usepackage{amssymb}
\usepackage{enumerate}
\usepackage{stfloats}
\usepackage{cases}
\usepackage{epstopdf}
\usepackage{soul,color} 

\soulregister\cite7
\soulregister\ref7
\soulregister\pageref7

\usepackage{tabularx}
\usepackage[table]{xcolor}
\usepackage{multirow}
\usepackage{booktabs}  
\usepackage{tabu} 

\usepackage[linesnumbered,ruled,vlined]{algorithm2e}

 \newcommand{\A}{\mathcal{A}}
 \newcommand{\B}{\mathcal{B}}
 \newcommand{\s}{\mathcal{S}}
 \newcommand{\M}{\mathcal{M}}
 \newcommand{\T}{\mathcal{T}}
 \newcommand{\K}{\mathcal{K}}

 \newcommand{\argmax}{\operatornamewithlimits{argmax}}
 \newcommand{\argmin}{\operatornamewithlimits{argmin}}

 \newcommand{\pbold}{\mathbf{p}}
 \newcommand{\pbolddown}{\widetilde{\mathbf{p}}}

 \newcommand{\gammai}{\gamma_i}
 \newcommand{\gammadown}{\widetilde{\gamma}}
 \newcommand{\gammaj}{\gamma_j}

 \newcommand{\Phibold}{\boldsymbol{\Phi}}

 \newcommand{\gammahat}{\mathit{\widehat{\gamma}}}

 \newcommand{\Fbold}{\mathbf{F}}

\newcommand{\gammabold}{\boldsymbol{\gamma}}
 \newcommand{\gammabolddown}{\widetilde{\boldsymbol{\gamma}}}
 \newcommand{\gammatargetbold}{\boldsymbol{\gamma}^{\mathrm{tar}}}

 \newcommand{\teta}[1]{\frac{\gamma_{#1}}{\gamma_{#1}+1}}
 \newcommand{\thetadown}[1]{\frac{\gammadown_{#1}}{\gammadown_{#1}+1}}

 \newcommand{\gammatargetdown}{\widetilde{\gamma}^{\mathrm{tar}}}
  
 \newcommand{\gammatarget}{\gamma^{\mathrm{tar}}}
 \newcommand{\hup}[2]{h_{#1 #2}}
 \newcommand{\hdown}[2]{\widetilde{h}_{#1 #2}}
 \newcommand{\Nup}[1]{N_{#1}}
 \newcommand{\Ndown}[1]{\widetilde{N}_{#1}}
 \newcommand{\pup}{p}
 \newcommand{\pdown}{\widetilde{p}}
 \newcommand{\Phidown}{\widetilde{\Phi}}
 \newcommand{\Pdownbold}{\widetilde{\mathbf{P}}}
 \newcommand{\Pdown}{\widetilde{P}}
\newcommand{\pmax}{p^{\mathrm{max}}}

\newcommand{\Pmaxdown}{\widetilde{P}^{\mathrm{max}}}

\newtheorem {corollary}{Corollary}

\newtheorem {proposition}{Proposition}
\newtheorem {remark}{Remark}
\newtheorem {definition}{Definition}

\begin{document}

\raggedbottom

\title{Low-Complexity SINR Feasibility Checking and Joint Power and Admission Control in Prioritized Multi-tier Cellular Networks 
}

\author{Mehdi~Monemi, Mehdi~Rasti, and Ekram Hossain
\thanks{M. Monemi is with the Dept. of Electrical and Computer Engineering, Neyriz Branch, Islamic Azad University, Neyriz, Iran (email: m\_monemi@shirazu.ac.ir). M. Rasti is with the Dept. of Computer Engineering and Information Technology, Amirkabir University of Technology, Tehran, Iran (email: rasti@aut.ac.ir).
E. Hossain is with the Dept. of Electrical and Computer Engineering, University of Manitoba, Canada (email: Ekram.Hossain@umanitoba.ca).}
}

\maketitle

	\begin{abstract}
		Next generation cellular networks will consist of  multiple tiers of cells and users associated with different network tiers may have different priorities (e.g., macrocell-picocell-femtocell networks with macro tier prioritized over pico tier, which is again prioritized over femto tier). Designing efficient joint power and admission control (JPAC) algorithms for such networks under a co-channel deployment (i.e., underlay) scenario is of significant importance. Feasibility checking of a given target signal-to-noise-plus-interference ratio (SINR) vector is generally the most significant contributor to the complexity of JPAC algorithms in single/multi-tier underlay cellular networks. This is generally accomplished through iterative strategies whose complexity is either unpredictable or of $O(M^3)$, when the well-known relationship between the SINR vector and the power vector is used, where $M$ is the number of users/links. In this paper, we derive a novel relationship between a given SINR vector and its corresponding uplink/downlink power vector based on which the feasibility checking can be performed with a complexity of $O(B^3+M B)$, where $B$ is the number of base stations. This is significantly less compared to $O(M^3)$ in many cellular wireless networks since the number of base stations is generally much lower than the number of users/links in such networks. The developed novel relationship between the SINR and power vector not only substantially reduces the complexity of designing JPAC algorithms, but also provides insights into developing efficient but low-complexity  power update strategies for prioritized multi-tier cellular networks. We propose two such algorithms and through simulations, we show that our proposed algorithms outperform the existing ones in prioritized cellular networks.
	\end{abstract}
\begin{keywords}
    5G cellular, multi-tier prioritized networks, underlay channel access, power and admission control, SINR assignment.
\end{keywords}
	

\thispagestyle{empty}

\section{Introduction}

To satisfy the ever-increasing demand from the new wireless applications and services such as 3D HD multimedia, VOIP, broadband internet services, HDTV, the fifth generation (5G) wireless communications technologies are being developed, which are expected to attain much higher mobile data volume per unit area, longer battery life, and reduced latency \cite{5G_metis}. 5G cellular wireless networks are expected to be a mixture of network tiers with different sizes, quality-of-service (QoS) requirements, transmit power levels, backhaul connections, and different radio access technologies.  Fig. \ref{fig:multi_tier} shows an instance of such a four-tier wireless network in which a macrocell with a wide coverage area coexists with several picocells and several in-house femtocells  together with a set of device-to-device (D2D) communication links.  In such a multi-tier network, a priority level may also be assigned to each network tier so that admission of users in some tiers is prioritized over that in other tiers. For example, the macro tier, where the base stations (BSs) are installed in a planned manner, may have a higher priority compared to the femto tier where the BSs could be installed in an unplanned manner by the users. Again,  the priority of the D2D tier may be lower than both the macro and pico tiers so that the D2D links do not cause QoS violations of the cellular links. As another example, a cellular network serving cognitive radios  may be considered as a prioritized two-tier network in which the primary radio network (PRN) serving the primary users (PUs) is the high-priority tier while the secondary network or the cognitive radio network (CRN) serving a set of secondary users (SUs) is the low-priority tier. Therefore,  admission of any of the SUs should not cause any QoS violation of any of the PUs. Depending on the operator's perspective,  different number of priority levels may be considered.

Prioritized wireless networks may be employed using either overlay or underlay dynamic spectrum access strategies. In the overlay spectrum access strategy, links are assigned with orthogonal channels (e.g., frequency bands). The channels which are unused by high-priority users are detected and exploited by low-priority users. In the underlay scenario, the entire frequency spectrum is shared by all of the users and thus the admission of each user causes interference to other users. Therefore, the interference caused by low-priority users must be controlled through power control strategies such that high-priority users are protected (i.e., achieve their target signal-to-interference-plus-noise ratios [target-SINRs]). In this paper, we consider an underlay system model wherein all users operate in a single shared channel.

Ideally, it is desirable to satisfy the QoS requirements (e.g., target-SINRs) of all users in the network. However, in an infeasible system, where all users may not be simultaneously supported with their target-SINRs, it is generally desirable to devise a joint power and admission control (JPAC) algorithm that protects the maximum number of users by considering their admission priority levels (e.g., in a prioritized two-tier CRN, the algorithm must protect all PUs, if possible, together with maximum number of admitted SUs). However, finding the maximum feasible set of prioritized users (i.e., the set  with maximum cardinality) is generally an NP-hard problem (\cite{12, andersin_gradual}). It requires an exhaustive search through all possible subsets of prioritized admitted users, leading to an unaffordable computational complexity in large-scale systems. Therefore, the existing algorithms look for sub-optimal solutions. 
For example, in \cite{andersin_gradual,rasti_topc,rasti_DFC,monemi_MTPC,single_tier_2015_lq}, several  JPAC algorithms are proposed to obtain sub-optimal solutions to the problem of finding the maximum feasible set of users in single-tier networks. In \cite{distributed_JPAC_antenna_arrays}, a distributed algorithm is introduced to minimize the total transmit power of primary and secondary links using antenna arrays. In \cite{JPAC_SSA1,ISMIRA,xing_dynamic_spectrum,LGRA,monemi_ESRPA,rasti_error_free,energy_feasibility_trade_off,rasti_distributed_uplink_2015,successiveGP1,Hoang_downlink_centralized_JPAC, JPAC_by_color_graph}, several centralized and distributed JPAC algorithms are proposed for two-tier CR networks to obtain sub-optimal solutions to the problem of finding the maximum feasible set of supported SUs subject to the constraint that all PUs are protected. Considering a higher priority for macrocell users, in \cite{femto_tow_tier_ngo_1,femto_tow_tier_ngo_2,hierarchical_tow_tier_guruacharya}, the problem of finding the maximum feasible set of femtocell  users is investigated in two-tier macrocell-femtocell networks.

There are two major issues related to the existing maximum feasible set JPAC algorithms in the literature. First, the computational complexities of the existing feasibility checking mechanisms  are large, and therefore, may not be suitable for large-scale networks (e.g., dense multi-tier cellular networks). Besides, the existing JPAC algorithms in the literature support a maximum of two priority levels. In this context, the main contributions of our paper can be stated as follows.
	
\begin{itemize}

\item The admission of a subset of users in the network is only possible if the corresponding power vector resulting in the target-SINRs of the admitted users is feasible (i.e., the corresponding transmit power of each link is non-negative and limited to the maximum allowed threshold). Feasibility checking of a given SINR vector is generally the most significant contributor to the complexity of JPAC algorithms in single/multi-tier underlay cellular networks. This is generally accomplished through iterative strategies (e.g., \cite{ISMIRA,rasti_error_free,successiveGP1,tuan_JPAC}) whose complexity is either unpredictable\footnote{The complexity of any iterative algorithm is related to the number of iterations required for that algorithm to converge. However, this is not known in advance or may not easily be calculated for iterative power control algorithms due to the variety of parameters that affect the convergence time of these algorithms. For example, path-gains for all users, noise powers, target-SINR values and even the precision of the convergence error affect the convergence time.}, or of $O(M^3)$, when the well-known relationship between the target SINR vector and the power vector (e.g., \cite{Hoang_downlink_centralized_JPAC, JPAC_by_color_graph,active_link_protection}) is used, where $M$ is the number of users/links. Recently a low-complexity centralized feasibility checking mechanism has been proposed in \cite{monemi_ESRPA}, however the derived relations are only applicable to a cognitive radio network having only one primary and one secondary base station. In this paper, we derive novel relationships between a given SINR vector and its corresponding uplink/downlink power vector, based on which the feasibility checking can be performed with a complexity of  $O(B^3+M B)$, where $B$ is the number of BSs. For many existing cellular networks at high traffic loads, this is considerably smaller than  $O(M^3)$ since each BS serves several users in its coverage area. Therefore, the developed relationships between the SINR and power vector  substantially reduce the complexity of designing JPAC algorithms, and it also provide insights into developing new low-complexity and efficient power update strategies for prioritized multi-tier cellular networks. 
		
\item  While there exist many JPAC algorithms in the literature for non-prioritized single-tier networks  (e.g., \cite{andersin_gradual,rasti_topc,rasti_DFC,monemi_MTPC,single_tier_2015_lq}), prioritized two-tier CRNs (e.g., \cite{distributed_JPAC_antenna_arrays,JPAC_SSA1,ISMIRA,xing_dynamic_spectrum,LGRA,monemi_ESRPA,rasti_error_free,energy_feasibility_trade_off,rasti_distributed_uplink_2015,successiveGP1,Hoang_downlink_centralized_JPAC, JPAC_by_color_graph}), and hierarchical two-tier macrocell-femtocell networks (e.g., \cite{femto_tow_tier_ngo_1,femto_tow_tier_ngo_2,hierarchical_tow_tier_guruacharya}), there exist very few research studies on prioritized multi-tier networks (e.g., see \cite{ekram_twoard_5G}). Based on the obtained  relationships between SINR and power vector (for uplink and downlink communication scenarios), we devise two efficient but low-complexity centralized JPAC algorithms for prioritized multi-tier cellular networks. To the best of our knowledge, these are the first algorithms proposed in the literature for finding the maximum feasible set of users for prioritized multi-tier networks supporting more than two priority levels. Regardless of the number of priority levels, the complexities of our proposed algorithms are far below those of  related existing algorithms. We show through simulations that the performance of our algorithms is superior to that of existing ones in terms of  average outage ratio of low-priority users (i.e., the ratio of the number of low-priority users who have not obtained their desired QoS to the total number of low-priority users). 

\end{itemize}
	
\begin{figure}
\centering
\includegraphics [width=254pt]{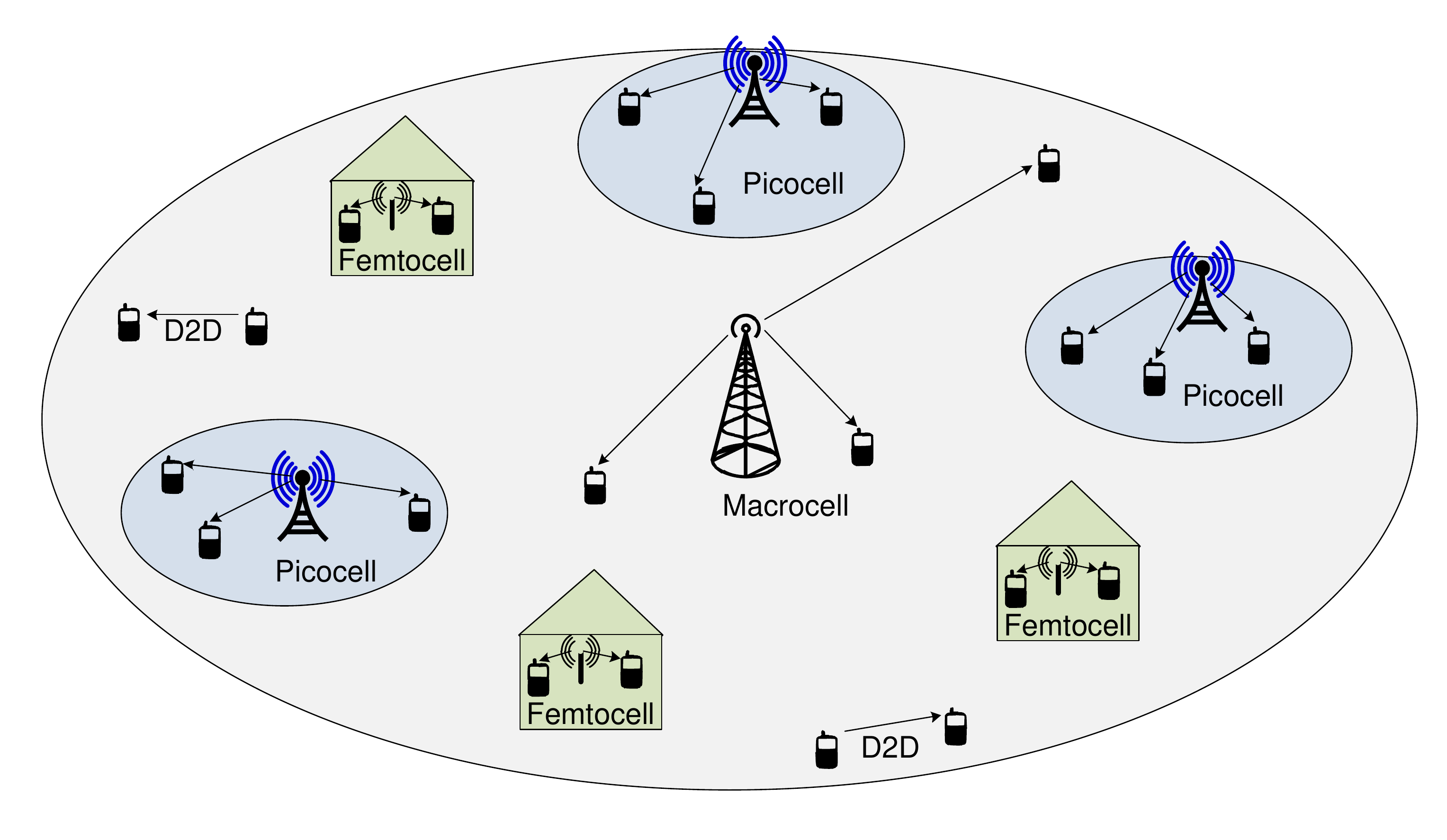} \\
\caption{A four-tier network consisting of a macrocell together with several picocells and femtocells and a set of D2D links.} \vspace{-10pt}
\label{fig:multi_tier}
\end{figure}
	
The rest of this paper is organized as follows. In Section \ref{sec:system_model} we describe the system model. A formal statement of the problem is given in Section \ref{sec:problem_statement}. The novel feasibility checking mechanism is obtained in Section \ref{sec:feasibility_checking} based on which our JPAC algorithms are proposed and studied in Section \ref{sec:proposed_algorithms}. Numerical results and conclusions are finally given in Sections \ref{sec:performance_evaluation_results} and \ref{sec:conclusions}, respectively.
	
\section{System Model and Assumptions}
	\label{sec:system_model}
	Consider underlay uplink/downlink transmission links in a prioritized multi-tier cellular wireless network consisting of $M$ users denoted by $\M=\{1,2,...,M\}$, $B$ BSs denoted by $\B=\{1,2,...,B\}$, $T$ tiers denoted by $\T=\{1,2,...,T\}$, and $K$ admission priority levels denoted by $\K=\{1,2,...,K\}$. Each user $i\in\M$ is assigned with a BS $b_i\in\B$, each BS $b\in\B$ belongs to a tier $t_b\in\T$ and each tier $t\in\T$ has an admission priority $k_t\in\K$ where $k_t=1$ and $k_t=K$ correspond to the case where tier $t$ is with the highest and lowest admission priorities, respectively. We assume that association of users to the BSs is such that it does not violate the backhaul constraints of the BSs. Also, we assume that each user is associated with only one network tier and all the users in a network tier have the same admission priority. In order not to lose the generality, we consider that a subset of tiers may be assigned to the same admission priority level; therefore, we have $K\leq T$, and thus for the case where each tier is assigned a unique admission priority, we have $K=T$. Also, let $\M_m^{\B}$, $\M_n^{\T}$, and $\M_q^{\K}$ denote the subset of users associated with BS $m\in\B$, tier $n\in\T$ and admission priority $q\in\K$, respectively, i.e.,
	\begin{align}
	\label{eq:3231}
		\M_m^{\B}=\{i\in\M |\  b_i=m \},
	\end{align}
	\begin{align}
	\label{eq:3232}
		\M_n^{\T}=\{i\in\M |\  t_m=n,\ \mathrm{where}\ m=b_i  \},
	\end{align}
	and
	\begin{align}
	\label{eq:3233}
		\M_q^{\K}=\{i\in\M |\ k_n=q, \ \mathrm{where}\ n=t_{b_i}  \}.
	\end{align}
	Similarly, let $\B_n^{\T}$  denote the subset of BSs associated with tier $n\in\T$ and $\B_q^{\K}$ denote the subset of BSs whose users have the admission priority $q\in\K$, i.e.,
	\begin{align}
	\label{eq:3234}
		\B_n^{\T}=\{b\in\B |\  t_b=n \},
	\end{align}
	and
	\begin{align}
	\label{eq:3235}
		\B_q^{\K}=\{b\in\B |\ k_{t_b}=q\}.
	\end{align}
	We model the wireless fading environment by large-scale path-loss and shadowing. The channels between different links experience independent fading and the network operates in a slow fading environment. 
	
\subsection{Uplink System Model}
	
	At any given snapshot in time, let $\pup_i$ be the transmit power of user $i$ and assume that $\hup{m}{i}$ denotes the uplink path-gain from user $i$ toward BS $m\in\B$ (e.g., $\hup{b_j}{i}$ is the uplink path-gain between user $i$ and the BS that user $j$ is associated with). The noise at BS $m$ is considered to be zero-mean additive white Gaussian whose power is denoted by $\Nup{m}$. The transmit power $p_i$ is always limited to a maximum value denoted by $\pmax_i$ (i.e., \mbox{$p_i\in[0, \pmax_i]$}). Considering the receivers to be conventional matched filters, for any given uplink transmit power vector $\mathbf{p}$ ($\pbold =\![p_1,p_2,...,p_M]^\mathrm{T}$), {the total uplink interference plus noise caused to user} $i$ {at its receiver is} ${I_i=\sum\limits_{\substack{j\in\M,  j \ne i}}{\!\!\hup{b_i}{j} p_j}+\Nup{b_i}}$ {and thus} the {\em normalized}\footnote{In a general radio transmission scenario (which includes spread-spectrum transmission), the actual SINR for each user is the normalized SINR multiplied by the processing gain.} uplink SINR of user $i$ at its BS denoted by $\gamma_i$ is
	
		\begin{equation}
		\label{eq:1}
			\gamma_i(\pbold)=
				\dfrac	{\hup{b_i}{i} p_i}
						{\sum\limits_{\substack{j\in\M, \\ j \ne i}}{\!\!\hup{b_i}{j} p_j}+\Nup{b_i}}, \ \forall i\in\M.
		\end{equation}
	Given an uplink SINR vector $\gammabold=[\gamma_1,\gamma_2,...,\gamma_M]^{\mathrm{T}}$, we can rewrite \eqref{eq:1} in matrix form and obtain  the corresponding power vector $\pbold$ as
		\begin{equation}
		\label{eq:40}
			\pbold(\gammabold)=\big(\mathbf{I}-\mathbf{F}(\gammabold)\big)^{-1}\mathbf{U}(\gammabold),
		\end{equation}
		where  $\mathbf{I}$ is an $M\times M$ identity matrix, and $\mathbf{U}(\gammabold)$ is an $M \times 1$ column vector wherein $U_i=\frac{\gamma_{i} \Nup{b_{i}}}{\hup{b_{i}}{i}}$ and $\mathbf{F}(\gammabold)$ is $M \times M$ matrix with  $F_{ij}=0$ for any $i=j$ and $F_{ij}=\frac{\gamma_i \hup{b_i}{j}}{\hup{b_i}{i}}$ for any $i \neq j$.
	
	\begin{definition}
		An uplink SINR vector $\gammabold$ is said to be feasible if it belongs to the set of feasible SINR vectors $\mathbf{F}_{\gammabold}$, where
		\begin{align}
		\label{eq:feasible_uplink_sinr}
			\Fbold_{\gammabold}=\{\gammabold | 0\leq p_i(\gammabold) \leq \pmax_i, \ \ \forall i\in\M \}.
		\end{align}
	\end{definition}

\subsection{Downlink System Model}
	
In the downstream communication scenario, at any given  snapshot in time, let $\pdown_i$ be the power level for the signal transmitted toward user $i$ by the BS serving user $i$ and assume that $\hdown{i}{m}$ denotes the downlink path-gain from BS $m\in\B$ toward user $i$ (e.g., $\hdown{i}{b_j}$ is the downlink path-gain between the  BS, that user $j$ is associated with, toward user $i$). 

	{Let }$\Pdown_m$ {be the total transmit power of BS }$m\in\B$, {i.e.},
	\begin{align}
	\label{eq:Pdown_m}
		{\Pdown_m=\sum\limits_{j\in\M_m^{\B}}{\pdown_j}}.
	\end{align}
	The aggregate transmission power of each BS $m\in\B$ is limited to a maximum threshold $\Pmaxdown_m$ (i.e., $\Pdown_m\leq \!\! \Pmaxdown_m, \forall m\in\B$). Consider $\Ndown{i}$ to be the noise power of user $i$, which is assumed to be zero-mean additive white Gaussian.	Considering the receivers to be conventional matched filters, for any given downlink transmit power vector $\pbolddown$ ($\pbolddown =\![\pdown_1,\pdown_2,...,\pdown_M]^\mathrm{T}$), {the total downlink interference plus noise at user $i$ is} ${\widetilde{I}_i=\sum_{m \ne b_i}{\!\!\hdown{i}{m} \Pdown_m}+\Ndown{i}=\sum\limits_{\substack{j\in\M, j \ne i}}{\!\!\hdown{i}{b_j} \pdown_j}+\Ndown{i}}$ and thus the downlink SINR of user $i$ denoted by $\gammadown_i$ is
	
	\begin{equation}
	\label{eq:1down}
		\gammadown_i(\pbolddown)=
			\dfrac	{\hdown{i}{b_i} \pdown_i}
					{\sum\limits_{\substack{j\in\M, \\ j \ne i}}{\!\!\hdown{i}{b_j} \pdown_j}+\Ndown{i}}, \ \forall i\in\M.
	\end{equation}
	For any given downlink SINR vector $\gammabolddown=[\gammadown_1,\gammadown_2,...,\gammadown_M]^{\mathrm{T}}$, the corresponding power vector $\pbolddown$ is obtained from \eqref{eq:1down} as:
	\begin{equation}
	\label{eq:40down}
		\pbolddown(\gammabolddown)=\big(\mathbf{I}-\widetilde{\mathbf{F}}(\gammabolddown)\big)^{-1}\widetilde{\mathbf{U}}(\gammabolddown),
	\end{equation}
	where  $\mathbf{I}$ is an $M\times M$ identity matrix, and $\widetilde{\mathbf{U}}(\gammabolddown)$ is  an $M \times 1$ column vector wherein $\widetilde{U}_i=\frac{\gammadown_{i} \Ndown{i}}{\hdown{i}{b_i}}$, and $\widetilde{\mathbf{F}}(\gammabolddown)$ is an $M \times M$ matrix with  $\widetilde{F}_{ij}=0$ for any $i=j$ and $\widetilde{F}_{ij}=\frac{\gammadown_i \hdown{i}{b_j}}{\hdown{i}{b_i}}$ for any $i \neq j$.
	
	\begin{definition}
		A downlink SINR vector $\gammabolddown$ is said to be feasible if it belongs to the set of feasible SINR vectors $\widetilde{\mathbf{F}}_{\gammabolddown}$, where
		\begin{multline}
		\label{eq:feasible_downlink_sinr}
			\widetilde{\Fbold}_{\gammabolddown}=  \{\gammabolddown | \ \pdown_i(\gammabolddown) \geq 0, \forall i \in\M \ \ \mathrm{and} \\  \sum_{i\in\M_m^{\B}}{\!\!\pdown_i(\gammabolddown)\leq \Pmaxdown_m, \ \forall m\in\B}
			\}.
		\end{multline}
	\end{definition}

\section{Problem Statement}
\label{sec:problem_statement}

In what follows, we first state the problem of uplink JPAC for prioritized multi-tier wireless networks and then we present a similar problem statement for the downlink network model. After stating the problems, we clarify how  the existing solutions lead to relatively high-complexity algorithms due to the complexity of feasibility checking of SINR vectors. Then we show how to reduce the complexity of JPAC algorithms through devising low-complexity feasibility checking mechanisms for uplink and downlink transmission scenarios in cellular networks.

\subsection{Uplink Case}
	
	For a given modulation scheme and a given maximum tolerable bit-error-rate (BER), the minimum allowed data rate for each user $i\in\M$ corresponds to a minimum acceptable SINR for that user (known as the uplink target-SINR denoted by $\gammatarget_i$). For any uplink power vector $\pbold$, the priority constraints force that if some user is provided with its target-SINR, all users having higher priorities are also provided with their target-SINRs. Let $\mathbf{G}_{\pbold}$ be the space of uplink power vectors for which the priority constraints hold, i.e.,
	\begin{multline}
	\label{eq:hp}
		\mathbf{G}_{\pbold} \!=\! \{\pbold | \mathrm{if}\ \exists i \in \! \M_q^{\K}\ \mathrm{for \ any}\ 1 < q \leq  K \ \mathrm{s.t.} \ \gammai(\pbold)\geq\gammatarget_i  
		\\ \mathrm{then} \
		\gammaj(\pbold)\geq\gammatarget_j\
		\mathrm{for\ each\ } q' < q \ \mathrm{and}\  j\in\M_{q'}^{\K}\}.
	\end{multline}	
	Given an uplink power vector $\pbold$, let $\s(\pbold)$ denote the set of users who attain their desired target-SINRs, i.e.,
	\begin{align}
	\label{eq:supported}
		\s(\pbold)\!=\!\left\{i\!\in\!\M|\gamma_i(\pbold)\geq \gammatarget_i\right\}.
	\end{align}	
	For any given target-SINR vector {$\gammatargetbold=[\gammatarget_1,\gammatarget_2,\dots,\gammatarget_M]^{\mathrm{T}}$}, it is desirable to obtain an uplink power vector $\pbold$ (where $0\leq p_i \leq \pmax_i, \forall i \in\M$) 
	for which the maximum possible number of users are provided with their target-SINRs while all the priority constraints hold as well. Therefore, we formally state the problem of finding the maximum feasible set of users (i.e., the feasible set with maximum cardinality) in uplink prioritized multi-tier wireless networks as the following optimization problem:
	\begin{align}
		\label{eq:opt1}
			\underset{\pbold}{\mathrm{maximize}}& \ \ \ |\s(\pbold)|  \ \ \   
			\nonumber 
		 \\
		 \mathrm{subject\ to} & \ \ \ 0\leq p_i \leq \pmax_i, \ \ \ \forall i\in\M,  
		 \nonumber
	  \\
		 & \ \ \ \pbold \in \mathbf{G}_{\pbold} ,
	\end{align}
	where $|.|$ denotes the cardinality of the corresponding vector. The first and second constraints correspond to the \textit{feasibility} of the power vector and \textit{priority} constraints, respectively.
			 
	The optimal solution to the above problem is not generally unique, i.e., there may be many transmit power vectors belonging to the set of solutions of (\ref{eq:opt1}). To show this, consider a system with two active users wherein user 1  is associated with a BS in a higher priority tier compared to that of user 2. Fig. \ref{fig:Region_Example} shows a 2-D space whose $x$ axis and $y$ axis are the transmit power levels of users 1 and 2, respectively. The solid lines are ${f_1(p_1,p_2)=\frac{p_1 h_{11}}{\gammahat_1} -  p_2 h_{12}-\sigma^2_1} $ and ${f_2(p_1,p_2)=\frac{p_2 h_{22}}{\gammahat_2} -  p_1 h_{21}-\sigma^2_2} $. Fig. \ref{fig:Region_Example}(a)  shows the case where both users can be supported with their desired target-SINRs (i.e., for all power vectors $\pbold$ located in the solid-filled region we have ${\gamma_1(\pbold)\geq \gammahat_1}$ and ${\gamma_2(\pbold)\geq \gammahat_2}$, or correspondingly, ${f_1(p_1,p_2)\geq 0}$ and ${f_2(p_1,p_2)\geq 0}$). In Fig. \ref{fig:Region_Example}(b), it is seen that it is not possible to serve both users simultaneously while satisfying their QoS requirements, however since user 1 is of higher priority, all feasible power vectors for which ${\gamma_1(\pbold)\geq \gammahat_1}$ (all points located in the solid-filled region of this figure) correspond to the solution of (\ref{eq:opt1}).
	
	\begin{figure}
			\centering
			\includegraphics [width=254pt]{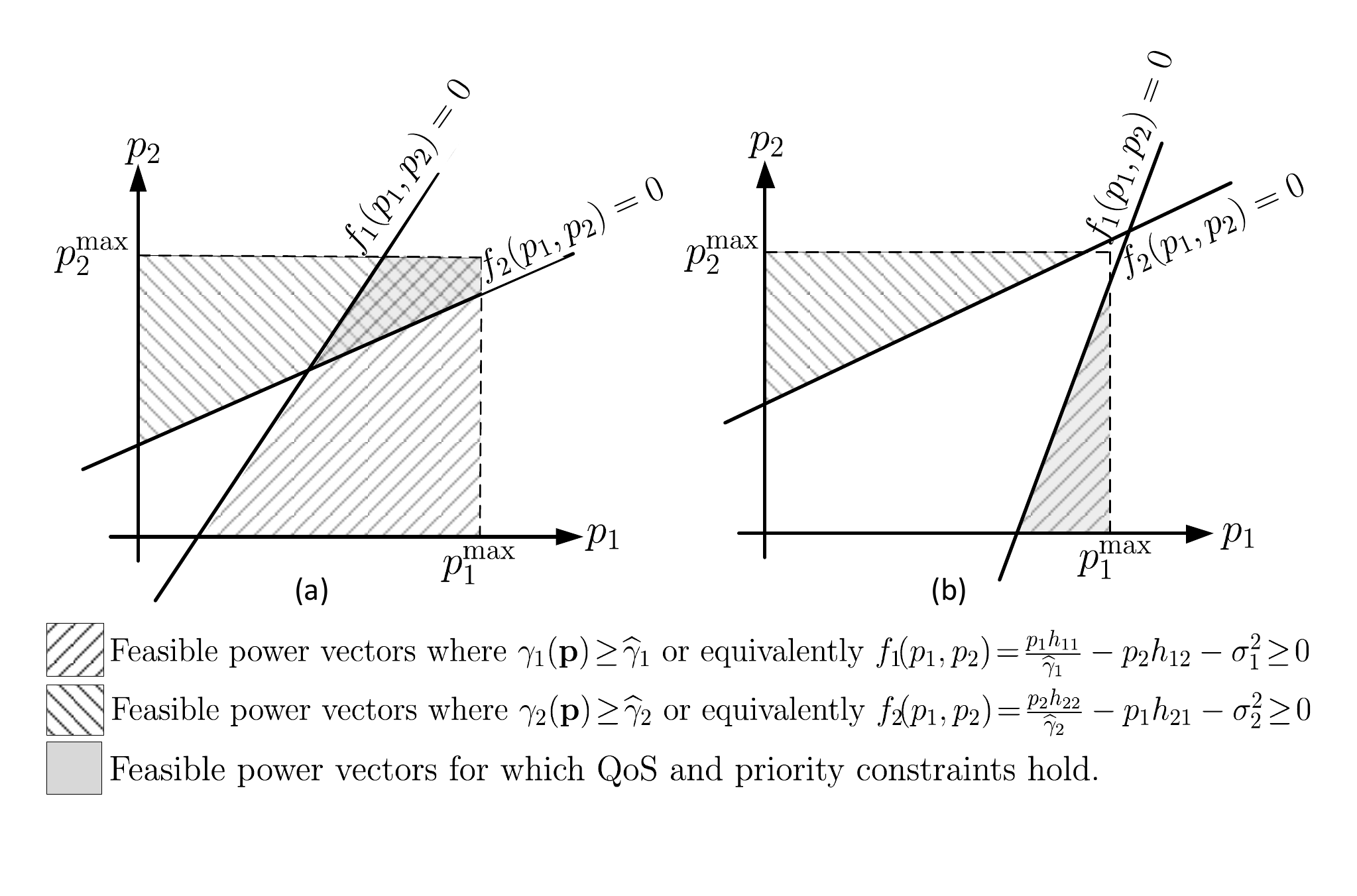}  \vspace{-30pt}
			\caption{System with two users where user 1 is of higher admission priority.} \vspace{-10pt}
			\label{fig:Region_Example}
		\end{figure}
		
	Among all power vectors corresponding to the solution of (\ref{eq:opt1}), we are specially interested in obtaining the ones that correspond to the minimum aggregate transmit powers of the users. Such optimal transmit power vectors ensure that supported users exactly meet their target-SINRs (i.e., equality in \eqref{eq:hp} holds for the second constraint in \eqref{eq:opt1}). To characterize these optimal solutions, the optimization problem in \eqref{eq:opt1} can be transformed into an SINR assignment problem as explained below.

	Consider the SINR assignment problem where each user $i$ may be assigned an SINR $\gammai\in\left\{\gammatarget_i,0\right\}$, where $\gammai=\gammatarget_i$ means that user $i$ is included in the set of active supported users, and $\gammai=0$  means that the potentially non-supported user $i$ is removed (i.e., its transmit power is zero) so as not to cause any extra interference to other users\footnote{The proposed admission control method is for service types for which a minimum acceptable target-SINR (or equivalently, target-data-rate) must be satisfied for all active users in the system.}. Let $\mathbf{T}_{\gammabold}$ be the space of SINR vectors that may be assigned to the users, i.e., 
	\begin{align}
	\label{eq:T_gamma}
		\mathbf{T}_{\gammabold}= \prod_{i\in\M}{\{\gammatarget_i,0\}},
	\end{align}
	and similar to {\eqref{eq:hp}}, let $\mathbf{G}_{\gammabold}$ be the space of SINR vectors for which the priority constraints hold, i.e.,
	\begin{multline}
	\label{eq:hgamma}
		\mathbf{G}_{\gammabold}=\{\gammabold | \mathrm{if}\ \exists i\in \M_q^{\K}\ \mathrm{for \ any}\ 1 < q \leq K \ \mathrm{s.t.} \ \gammai\geq\gammatarget_i  
		\\ \mathrm{then} \
		\gammaj\geq\gammatarget_j\
		\mathrm{for\ each\ } q' < q\ \mathrm{and}\  j\in\M_{q'}^{\K}\}.
	\end{multline}
	Also, let $\A(\gammabold)\subset\M$ be the subset of users assigned with  non-zero SINR values referred to as the admitted users (i.e., $\A(\gammabold)=\{i\in\M | \gamma_i>0\}$).
	We modify the power allocation problem in \eqref{eq:opt1} to the following SINR assignment and admission control problem:	
	\begin{align}
	\label{eq:opt_admission}
		\underset{\gammabold\in \mathbf{T}_{\gammabold}}{\mathrm{maximize}}& \ \ \ |\A(\gammabold)|
					&  \nonumber \\
						\mathrm{subject\ to }&\ \ \ \gammabold\in \mathbf{F}_{\gammabold} \cap \mathbf{G}_{\gammabold},
	\end{align}
	where $\mathbf{F}_{\gammabold}$ and $\mathbf{G}_{\gammabold}$ are given in \eqref{eq:feasible_uplink_sinr} and \eqref{eq:hgamma}, respectively.

	\begin{proposition}
	\label{th:1}
		The SINR vector obtained from \eqref{eq:opt_admission} corresponds to a minimal\footnote{A minimal power vector is the one among feasible power vectors that is not greater (component-wise) than any other feasible power vector.} power vector of the set of solutions of \eqref{eq:opt1}. 
			\end{proposition}
	\begin{proof}
		It can be proved by taking the steps similar to	those for the proof of Proposition 1 in \cite{monemi_ESRPA}.
	\end{proof}

	Proposition \ref{th:1} states that the optimal SINR vector given by the optimization problem in \eqref{eq:opt_admission} results in the minimum aggregate transmit power required for the maximum feasible set of prioritized users obtained from the optimization problem in \eqref{eq:opt1}.
	
\subsection{Downlink Case}
	Let $\gammatargetdown_i$ be the downlink minimum acceptable SINR (downlink target-SINR) of user $i$.
	Similar to \eqref{eq:opt_admission}, the problem of SINR assignment and admission control for finding the maximum feasible set of prioritized users in multi-tier downlink communication model is stated as
	\begin{align}
	\label{eq:opt_admissiondown}
		\underset{\gammabolddown\in \mathbf{T}_{\gammabolddown}}{\mathrm{maximize}}& \ \ \ |\A(\gammabolddown)|
					&  \nonumber \\
						\mathrm{subject\ to }&\ \ \ \gammabolddown\in \widetilde{\mathbf{F}}_{\gammabolddown} \cap \widetilde{\mathbf{G}}_{\gammabolddown},
	\end{align}
	where $\A(\gammabolddown)=\{i\in\M | \gammadown_i>0\}$ is the set of admitted downlink connections, $\mathbf{T}_{\gammabolddown}=\prod_{i\in\M}{\{\gammatargetdown_i,0\}}$,  $\widetilde{\mathbf{F}}_{\gammabolddown}$ is given by \eqref{eq:feasible_downlink_sinr},  and  $\widetilde{\mathbf{G}}_{\gammabolddown}$ represents the space of downlink SINR vectors for which the priority constraints hold, i.e.,

	\begin{multline}
	\label{eq:hthetadown}
		\widetilde{\mathbf{G}}_{\gammabolddown}=\{\gammabolddown | \mathrm{if}\ \exists i\in \M_q^{\K}\ \mathrm{for \ any}\ 1 < q \leq K \ \mathrm{s.t.} \ \gammadown_i\geq\gammatargetdown_i  
		\\ \mathrm{then} \
		\gammadown_j\geq\gammatargetdown_j\
		\mathrm{for\ each\ } q' < q \ \mathrm{and}\  j\in\M_{q'}^{\K}\}.
	\end{multline}

	In what follows, we construct novel relationships between a power vector and its corresponding SINR vector for devising low-complexity feasibility checking mechanisms for both uplink and downlink communication models. Then, using the derived relationships, we devise novel admission metrics and propose two efficient but low-complexity JPAC algorithms to solve each of the problems in \eqref{eq:opt_admission} and \eqref{eq:opt_admissiondown} for prioritized multi-tier cellular wireless networks for both uplink and downlink communications.
	
\section{Low-Complexity Feasibility Checking Mechanism}
\label{sec:feasibility_checking}
	
	The problem of  finding the maximum feasible set of users for admission control with/without the priority constraints is generally NP-hard  (\cite{andersin_gradual,JPAC_adhoc_convex_relaxation_single_PU,p_norm_proof_for_NP_Hardness}), and needs exhaustive search through all possible SINR vectors $\gammabold\in \mathbf{T}_{\gammabold}$. Therefore, practical JPAC algorithms generally obtain a sub-optimal solution of the maximum feasible set problem by iteratively removing the unsupported users according to some admission metric and checking the feasibility of the remaining set of users at each iteration until the system becomes feasible. The process of checking the feasibility of admitted users at each iteration is generally the most significant contributor to the complexity of all such algorithms. This is mostly accomplished by checking the feasibility of the power vector corresponding to the desired SINR vector through (\ref{eq:40}) or (\ref{eq:40down}) which needs matrix inversion with the complexity of $O(M^3)$. Another way is to obtain the equilibrium power vector from the iterative constrained target-SINR tracking power control (TPC) algorithm proposed in \cite{foschini_base_TPC,grandhi_base_TPC,grandhi_zander_base_constrained_TPC} whose complexity and speed of convergence are not known in advance and depends on the network parameters.
		
	In what follows, for the uplink scenario, wherein the transmit power of each user is limited to a maximum threshold, we first obtain a novel and low-complexity relationship that maps the uplink SINRs of the users to the corresponding transmit powers of the users based on which a low-complexity feasibility checking mechanism for cellular networks is devised for the uplink system model.
	On the other hand, the downlink total transmit power of each BS must be limited. Therefore, for the downlink scenario, we obtain a similar novel relationship between the downlink SINRs of the users and the total downlink transmit powers of the BSs. The derived relations can be employed in the JPAC algorithms instead of the feasibility checking mechanisms in (\ref{eq:40}) or (\ref{eq:40down}) to  decrease the complexity of such algorithms substantially.

\subsection{Uplink Low-Complexity Feasibility Checking}
\vspace{0.2cm}
	\begin{proposition}
	\label{prop:2}
		Given an uplink SINR vector $\gammabold=[\gamma_1, \gamma_2, ..., \gamma_{M}]^{\mathrm{T}}$,  
				the corresponding power vector $\pbold$ that results in $\gammabold$ 
		 is obtained from 
		\begin{align}
		\label{eq:p_i}
			p_i= \teta{i}\dfrac{\Phi_{b_i}}{\hup{b_i}{i}} , \ \ \ \ \ \forall i\in\M,
		\end{align}
		where $\Phibold=[\Phi_{1}, \Phi_{2}, ..., \Phi_{B}]^{\mathrm{T}}$ is obtained as
		\begin{align}
		\label{eq:phi}
			\Phibold(\gammabold)=\left(\mathbf{I-H(\gammabold)}\right)^{-1} \mathbf{N},
		\end{align}
		in which $\mathbf{I}$ is a $B\times B$ identity matrix, $\mathbf{N}=[N_{1}, N_{2}, ..., N_{B}]^{\mathrm{T}}$, and $\mathbf{H}$ is a $B\times B$ matrix whose elements are obtained by
		\begin{align}
		\label{eq:H_elements}
			H_{mn}=
			\begin{cases}
			\sum\limits_{i\in\M_m^{\B}}{\!\! \teta{i}}, & \mathrm{if\ } m=n, \\
			\sum\limits_{i\in\M_n^{\B}}{\!\! \frac{\hup{m}{i}}{\hup{n}{i}} \teta{i}}, & \mathrm{if\ } m\neq n.
			\end{cases}
		\end{align}
	\end{proposition}
	\begin{proof}
		See \textbf{Appendix {\ref{apx:1}}}.
	\end{proof}
	
	\begin{corollary}
		An uplink SINR vector $\gammabold$ 
		is feasible if and only if
		\begin{align}
		\label{eq:fisible_phi}
			0 \leq \Phi_m(\gammabold) \leq   \Phi^{\mathrm{max}}_m(\gammabold) , \ \ \forall m\in\B,
		\end{align}
		where $\Phi_m(\gammabold)$ is obtained from \eqref{eq:phi} and $\Phi^{\mathrm{max}}_m(\gammabold)= \min\limits_{i\in\M_m^{\B}} {\{ \frac{p_i^{\mathrm{max} } \hup{m}{i}(\gamma_i + 1) }  {\gamma_i} \} }$.
	\end{corollary}

\subsection{Downlink Low-Complexity Feasibility Checking}

	\begin{proposition}
	\label{prop:3}
		Given a downlink SINR vector  $\gammabolddown=[\gammadown_1, \gammadown_2, ..., \gammadown_{M}]^{\mathrm{T}}$, 
		the corresponding power vector $\pbolddown$ that results in $\gammabolddown$ 
		is obtained from 
		\begin{align}
		\label{eq:p_idown}
			\pdown_i=  \frac{1}{\hdown{i}{b_i}} \thetadown{i} \times \left( \sum\limits_{n\in\B}{ \hdown{i}{n}\Pdown_n } + \Ndown{i} \right),
		\end{align}
		where $\Pdownbold=[\Pdown_{1}, \Pdown_{2}, ..., \Pdown_{B}]^{\mathrm{T}}$ is obtained as
		\begin{align}
		\label{eq:Pdown}
			\Pdownbold(\gammabolddown)=\left(\mathbf{I-\widetilde{H}(\gammabolddown)}\right)^{-1} \widetilde{\mathbf{N}}^{*},
		\end{align}
		in which $\mathbf{I}$ is a $B\times B$ identity matrix, $\widetilde{\mathbf{N}}^{*}=[\Ndown{1}^{*}, \Ndown{2}^{*}, ..., \Ndown{B}^{*}]^{\mathrm{T}}$, where
		\begin{align}
			\Ndown{m}^{*}=\sum\limits_{i\in\M_m^{\B}}{\dfrac{1}{\hdown{i}{m}}\thetadown{i} \Ndown{i} },
		\end{align}
		and $\widetilde{\mathbf{H}}$ is a $B\times B$ matrix whose elements are obtained from
		\begin{align}
		\label{eq:H_elements_down}
			\widetilde{H}_{mn}=
			\begin{cases}
			\sum\limits_{i\in\M_m^{\B}}{\!\! \thetadown{i}}, & \mathrm{if\ } m=n, \\
			\sum\limits_{i\in\M_n^{\B}}{\!\! \frac{\hdown{i}{n}}{\hdown{i}{m}} \thetadown{i}}, & \mathrm{if\ } m\neq n.
			\end{cases}
		\end{align}
	\end{proposition}
	\begin{proof}
		See \textbf{Appendix {\ref{apx:2}}}.
	\end{proof}

	\begin{corollary}
		A downlink SINR vector $\gammabolddown$ 
		is feasible if and only if
		\begin{align}
		\label{eq:fisible_simple_down}
			0 \leq \Pdown_m(\gammabold)\leq \Pdown_m^{\mathrm{max}}, \ \ \ \forall m\in\B,
		\end{align}
		where $\Pdown_m(\gammabold)$ is obtained from \eqref{eq:Pdown}. 
	\end{corollary}
	
	\begin{remark}
		The computational complexities of calculating $(\mathbf{I-H(\gammabold)})$ in \eqref{eq:phi} from \eqref{eq:H_elements} and  $(\mathbf{I-\widetilde{H}(\gammabolddown)})$ in \eqref{eq:Pdown} from \eqref{eq:H_elements_down} are both of $O(M\!\times\!B)$. Thus, together with the complexity of the inversion of the corresponding matrices, the overall complexity of feasibility checking of a given uplink or downlink SINR vector through \eqref{eq:fisible_phi} or \eqref{eq:fisible_simple_down}, respectively, is only of $O(B^3+M\!\times\! B)$. This is significantly smaller compared to the complexity of the traditional feasibility checking mechanisms through \eqref{eq:40} or \eqref{eq:40down} which is of $O(M^3)$. This is because,  the number of users ($M$) is generally  much higher than the number of BSs ($B$). Considering the fact that feasibility checking is generally the most significant contributor to the complexity of JPAC algorithms in traditional non-prioritized single-tier or prioritized multi-tier cellular networks, by applying our proposed feasibility checking mechanism,  the complexities of  the existing JPAC algorithms can be significantly reduced. Also, the proposed mechanism provides insights into developing efficient JPAC algorithms for multi-tier wireless networks.
	\end{remark}

\section{Proposed JPAC Algorithms and their Properties}
	\label{sec:proposed_algorithms}

	In what follows, based on the low-complexity SINR-to-power relations obtained in the previous section, we propose two low-complexity and efficient JPAC algorithms, namely, the \textbf{M}ulti-tier \textbf{E}ffective \textbf{S}tepwise user removal with high-priority user \textbf{P}rotection \textbf{A}lgorithm (MESPA), and \textbf{M}ulti-tier \textbf{L}ow-complexity \textbf{S}tepwise user removal with high-priority user \textbf{P}rotection \textbf{A}lgorithm (MLSPA). The properties of the proposed algorithms are also investigated. With the objective of maximizing the number of supported users, the key idea in both algorithms is to devise low-complexity removal metrics based on the obtained SINR-to-power relations. The removal metrics are then employed in the proposed JPAC algorithms by iteratively finding and removing the candidates according to the removal metrics and checking the feasibility of the system through the proposed low-complexity feasibility checking mechanisms until the system becomes feasible. In finding the removal candidates, we first find the BS with maximum infeasibility measure and then, among the users served by that BS, we find the user whose removal results in the maximal decrease in the infeasibility measure and approximal infeasibility measure of that BS in the MESPA and MLSPA, respectively, as formally described in the following subsection.
	
\subsection{Proposed JPAC Algorithms}
	In what follows, we first present the MESPA and then we revise the admission metric of MESPA to form a new admission metric with lower complexity and propose the MLSPA. The lower complexity of MLSPA leads to a slight degradation in network performance compared to that for MESPA as will be shown in the numerical results.
	\subsubsection{MESPA}
	Assume that all users in the network are initially admitted and assigned with their target-SINRs (i.e., $\A\leftarrow \M$ and $\gamma_i=\gammatarget_i$ for the uplink and $\gammadown_i=\gammatargetdown_i$ for the downlink communication scenarios for all $i\in\M$).  Simultaneous admission of all users may result in the infeasibility of the system. For such a scenario, we  devise a low-complexity iterative strategy through which at each step, some users are removed (aiming at a small number of removals) from the set of active users until the remaining set of active users results in a feasible prioritized system. The process of determining the removal candidate for an infeasible system at each step is accomplished in two phases: determination of the candidate BS/BSs from which some users have to be removed and then finding the removal candidate among the users associated with the corresponding BS/BSs.
	
	The BS/BSs from which some users have to be removed is/are determined as the one/ones with the lowest priority which has/have at-least one active user. Suppose that BSs $\B_q^{\K}\subset\B$ with the priority level $q$ have been chosen as such BSs. In what follows we describe how the removal candidate is determined from $\M_q^{\K}$ where $\M_q^{\K}\subset\M$ is the subset of users having the admission priority of $q\in\K$ associated with BSs $\B_q^{\K}$.  
	
	Let $\gammabold_{\!\A}$ and $\gammabolddown_{\!\A}$ denote the SINRs of the users when all admitted users (any user $i\in\A$) are assigned with their target-SINRs and other removed users are assigned with zero SINR in the uplink and downlink communication scenarios, respectively. Also, assume that $\gammabold_{\!\A}$ and $\gammabolddown_{\!\A}$ are infeasible SINR vectors. This means that \eqref{eq:fisible_phi} and \eqref{eq:fisible_simple_down} do not hold for at least one BS $m\in\B$ in the uplink and downlink scenarios, respectively. Let ${n}^{*}\in\B$ and $\widetilde{n}^{*}\in\B$ be the BSs with maximum infeasibility measure for the uplink and downlink communications, respectively, defined as follows:
	\begin{align}
	\label{eq:ncandid}
		n^{*}\!\! = \!
		\begin{cases}
		 	\!\!\!\!\!\!\!\!\!\!\!\!
		 	\argmax\limits_{\ \ \ \ \ \ n\in\B | \Phi_n\!(\gammabold_{\!\A})<0}{ \!\!\!\!\!\!\!\!\!\!\! \{  \Phi_n(\gammabold_{\!\A})    \} }, 
		 	\hspace{20pt}
		 	\mathrm{if}\ \exists  n \! \in \! \B \ \mathrm{s.t.}\ \Phi_n(\gammabold_{\!\A}) \! < \! 0
		 	 \\
		 	\hspace{-46pt}
		 	\argmax\limits_{\ \ \ \ \ \ \ \ \ \ \ \ \ n\in\B | \Phi_n\!(\gammabold_{\!\A})>\Phi_n^{\mathrm{max}}(\gammabold)} {\!\!\!\!\!\!\!\!\!\!\!\!\!\!\!\!\!\!\!\!\!\!\!\!\! \{ \Phi_n(\gammabold_{\!\A}) \! - \! { {\Phi}_n^{\mathrm{max}}}(\gammabold)  \} }, 
		 	\hspace{20pt}
		 	\mathrm{otherwise},
		\end{cases}
	\end{align}
	and 
	\begin{align}
	\label{eq:ncandid_down}
		\widetilde{n}^{*}\! = \!
		\begin{cases}
			\!\!\!\!\!\!\!\!\!\!\!\!
		 	\argmax\limits_{\ \ \ \ \ \ n\in\B | \Pdown_n\!(\gammabolddown_{\!\A})<0}{ \!\!\!\!\!\!\!\!\!\!\! \{  \Pdown_n(\gammabolddown_{\!\A})    \} }, 
		 	\hspace{20pt}
		 	\mathrm{if}\ \exists  n \! \in \! \B \ \mathrm{s.t.}\ \Pdown_n(\gammabolddown_{\!\A}) \! < \! 0
		 	 \\
		 	\!\!\!\!\!\!\!\!\!\!\!\!\!\!\!\!\!\!\!\!\!\!
		 	\argmax\limits_{\ \ \ \ \ \ \ \ \ \ \ n\in\B | \Pdown_n\!(\gammabolddown_{\!\A})>\Pdown_n^{\mathrm{max}}} {\!\!\!\!\!\!\!\!\!\!\!\!\!\!\!\!\!\!\!\! \{ \Pdown_n(\gammabolddown_{\!\A}) \! - \! { {\Pdown}_n^{\mathrm{max}}}  \} }, 
		 	\hspace{20pt}
		 	\mathrm{otherwise}.
		\end{cases}
	\end{align}
	Note that $n^{*}$ and $\widetilde{n}^{*}$ in \eqref{eq:ncandid} and \eqref{eq:ncandid_down}, respectively, are first selected from the BSs for whom the lower-bound feasibility condition is violated (i.e., BSs $m$ for which \mbox{$\Phi_m<0$} or \mbox{$\Pdown_m<0$}) and therefore BSs with lower-bound infeasibility are supposed to have more degree of infeasibility as compared with those with upper-bound infeasibility (i.e., BSs $m$ for which \mbox{$\Phi_m>\Phi_m^{\mathrm{max}}$} or \mbox{$\Pdown_m>\Pdown_m^{\mathrm{max}}$}). This is because, for a given BS $m\in\B$ with lower-bound infeasibility, successive removals of active users first leads to the lower-bound feasibility (\mbox{$\Phi_m \geq 0$} or \mbox{$\Pdown_m \geq 0$}) and then results in the upper-bound feasibility (\mbox{$\Phi_m\leq\Phi_m^{\mathrm{max}}$} or \mbox{$\Pdown_m\leq\Pdown_m^{\mathrm{max}}$}).
		
	At each removal iteration, we may determine the removal candidate $i^{*}\in\M_q^{\K}\cap \A$ or $\widetilde{i}^{*}\in\M_q^{\K}\cap\A$ as the one whose removal leads to the maximal decrease in the infeasibility measure of the BSs ${n}^{*}\in\B$ or $\widetilde{n}^{*}\in\B$ in the uplink and downlink communications, respectively, i.e.,
	
	\begin{align}
	\label{eq:removal_candidate_up}
		i^{*}\!\! = \!\!
		\begin{cases}
			\hspace{-106pt}
			 	\argmin\limits_{\hspace{110pt} i\in\M_q^{\K}\cap \A | \Phi_{n^{\!*}}\!(\gammabold_{\!\A\setminus\{i\}}\!)>
			 	\Phi_{n^{\!*}}^{\mathrm{max}}(\gammabold_{\!\A\setminus\{i\}}\!) } 
			 	{
			 	\hspace{-105pt}
			 		\{ \Phi_{n^{\!*}}(\gammabold_{\!\A\setminus \{i\} })  - \Phi_{n^{\!*}}^{\mathrm{max}}(\gammabold_{\!\A\setminus\{i\}}\!)   \} }, \ \ \ \ 
			 		\\
			 		\hspace{45pt} 
			 		\mathrm{if}
			 		\ \exists i\!\in\!\M_{n^{\!*}}^{\B} \ \mathrm{s.t.} \ \Phi_{\!n^{\!*}}(\gammabold_{\!\A\setminus\{i\}}) \!> \! \Phi_{n^{\!*}}^{\mathrm{max}}(\gammabold_{\!\A\setminus\{i\}}\!)
				\\
			\hspace{-62pt}
				 	\argmin\limits_{\ \ \ \ \ \ \ \ \ \ \ \ \ \ \ \ \ \ i\in\M_q^{\K}\cap\A | \Phi_{n^{\!*}}(\gammabold_{\!\A\setminus\{i\}})<0} \!\! {\!\!\!\!\!\!\!\!\!\!\!\!\!\!\!\!\!\!\!\!\!\!\!\!\!\!\!\!\!\!\!\!\! 
				 		\{ \Phi_{n^{\!*}}(\gammabold_{\!\A \setminus \{i\} }\! )  \} }, 
				 	\hspace{85pt} \mathrm{otherwise}.
		\end{cases}
	\end{align}
	and
	\begin{align}
	\label{eq:removal_candidate_down}
		\widetilde{i}^{*}\!\! = \!\!
		\begin{cases}
			\hspace{-78pt}
			 	\argmin\limits_{\hspace{70pt} \ \ \ \ i\in\M_q^{\K}\cap \A | \Pdown_{\widetilde{n}^{\!*}}\!(\gammabolddown_{\!\A\setminus\{i\}}\!)>
			 	\Pdown_{n^{\!*}}^{\mathrm{max}} } 
			 	{
			 	\hspace{-80pt}
			 		\{ \Pdown_{\widetilde{n}^{\!*}}(\gammabolddown_{\!\A\setminus \{i\} }) - \Pdown_{\widetilde{n}^{\!*}}^{\mathrm{max}})   \} }, \ \ \
			 		\\
	 				\hspace{70pt} 
	 				\mathrm{if}\ \exists i\!\in\!\M_{\widetilde{n}^{\!*}}^{\B} \ \mathrm{s.t.} \ \Pdown_{\widetilde{n}^{\!*}}(\gammabolddown_{\!\A\setminus\{i\}}) \!> \! \Pdown_{\widetilde{n}^{\!*}}^{\mathrm{max}}
				\\
			\hspace{-63pt}
				 	\argmin\limits_{\ \ \ \ \ \ \ \ \ \ \ \ \ \ \ \ \ \ \ i\in\M_q^{\K}\cap\A | \Pdown_{\widetilde{n}^{\!*}}(\gammabolddown_{\!\A\setminus\{i\}})<0} \!\! {\!\!\!\!\!\!\!\!\!\!\!\!\!\!\!\!\!\!\!\!\!\!\!\!\!\!\!\!\!\!\!\!\! 
				 		\{ \Pdown_{\widetilde{n}^{\!*}}(\gammabolddown_{\!\A \setminus \{i\} }\! )  \} }, 
				 	\hspace{80pt}
				 	 \mathrm{otherwise}.
		\end{cases}
	\end{align}

	Based on what has been discussed so far, the MESPA JPAC for prioritized multi-tier wireless networks is proposed in \textbf{Algorithm 1}.

\begin{algorithm}

    \SetAlgoNoLine    
    \DontPrintSemicolon
    
    \small\textbf{Initialization:}\normalsize\;
    \Indp 
     	
     	Let all users be initially admitted (i.e., \mbox{$\A\leftarrow \M$}) and assigned with their target-SINRs.\;
     	
     	Let $q \leftarrow K$ to start removing users from the cells having the lowest priority.\;
     	
    \Indm
    	\small \textbf{Admission Control:} \normalsize\;
    \Indp
    
    Calculate $\Phibold(\gammabold_{\!\A})$ from \eqref{eq:phi} for the uplink case or $\Pdownbold(\gammabolddown_{\!\A})$ from \eqref{eq:Pdown} for the downlink case.\;
    
    \If {$\gammabold_{\!\A}$ for the uplink case is infeasible (checked through \eqref{eq:fisible_phi}) or $\gammabolddown_{\!\A}$ for the downlink case is infeasible (checked through \eqref{eq:fisible_simple_down})}
    {
    	\While{ $|M_q^{\K}\cap \A|=0$}
    	{
    		\eIf{$q>1$}
    		{
    			$q\leftarrow q-1$\;
    		}
    		{
    			None of the users can be supported,
    			terminate the algorithm.
    		}
    	}

   		Determine $n^{*}\in\B$ for the uplink case from \eqref{eq:ncandid} or $\widetilde{n}^{*}\in\B$ for the downlink case from \eqref{eq:ncandid_down}.\; 
   		
		Let \mbox{$\A\leftarrow\A \setminus \{i^{*}\}$} for the uplink case, where $i^{*}$ is obtained by \eqref{eq:removal_candidate_up}, or \mbox{$\A\leftarrow\A \setminus \{\widetilde{i}^{*}\}$} for the downlink case where $\widetilde{i}^{*}$ is obtained by \eqref{eq:removal_candidate_down}\;
   		
    }
    \Indm
    \small \textbf{Power Calculation:} \normalsize\;
    \Indp
 
 	Calculate the  power vector corresponding to $\gammabold_{\!\A}$ for the uplink case using \eqref{eq:p_i} or the power vector corresponding to $\gammabolddown_{\!\A}$ for the downlink case using \eqref{eq:p_idown}.\;
\caption{MESPA}
\end{algorithm}

	It is seen from step 7 of MESPA and from \eqref{eq:removal_candidate_up} and \eqref{eq:removal_candidate_down} that the determination of the removal candidate at each iteration needs   $|\M_q^{\K}\cap\A|$ times of matrix inversions. In what follows, we propose another algorithm with less complexity which needs only one matrix inversion in the calculation of the removal candidate at each iteration. We first consider the uplink scenario and then extend the results for the downlink scenario as well.

\subsubsection{MLSPA}
	Similar to the previous algorithm, consider that all users in the uplink communication are initially admitted and also assume that at each iteration, $n^{*}\in\B$ is determined from \eqref{eq:ncandid}. We determine the removal candidate $i^{*}$ from the set of users associated with BSs $\B_q^{\K}$ as the one with maximum approximal decrease in the infeasibility measure of BS $n^{*}$. To do this, we use the first-order sensitivity analysis of matrix $\Phibold$. From \eqref{eq:phi}, let $\mathbf{A=I-H}(\gammabold_{\!\A})$  and $\Phibold=\mathbf{A}^{-1} \mathbf{N}$. The partial derivative of $\Phi_{n^{*}}$ with respect to the coefficients of $\mathbf{A}$ is obtained as follows \cite{sensativity_analysis}:
	\begin{align}
		\frac{\partial \Phi_{n^{*}}}{\partial A_{mn}}= - (\mathbf{A}^{-1})_{n^{\!*}m} \Phi_n.
	\end{align}
	Therefore, by using the first-order Taylor series expansion, the variation of $\Phi_{n^{*}}$ (denoted by $\Delta \Phi_{n^{*}}$) with respect to the marginal variation of the coefficients of $\mathbf{A}$ (denoted by $\Delta A_{mn}$) is expressed as
	\begin{align}
	\label{eq:Delta_phi_n}
		\Delta \Phi_{n^{*}} = \sum\limits_{m,n\in\B} -(\mathbf{A}^{-1})_{n^{\!*}m} \Phi_{n} \Delta A_{mn}.
	\end{align}
	Let $\Delta \Phi_{n^{*}}^{i}$ denote the variation of $\Delta \Phi_{n^{*}}$ when user $i\in\M$ is removed from the set of admitted users. From \eqref{eq:H_elements} and \eqref{eq:Delta_phi_n} we have
	\begin{align}
	\label{eq:Delta_phi_n_i}
		\Delta \Phi_{n^{*}}^{i} = \sum\limits_{m,n\in\B} -(\mathbf{A}^{-1})_{n^{\!*}m} \Phi_{n} \Delta A_{mn}^{i},
	\end{align}
	where
	\begin{align}
	\label{eq:Delta_A}
		\Delta A_{mn}^{i}=
		\begin{cases}
			- \teta{i}, & \mathrm{if\ } m=n\ \mathrm{and}\ i\in\M_m^{\B}, \\
			- \frac{\hup{m}{i}}{\hup{n}{i}} \teta{i}, & \mathrm{if\ } m\neq n\ \mathrm{and}\ i\in\M_n^{\B},\\
			0, & \mathrm{else}.
		\end{cases}
	\end{align}
	It can be verified from \eqref{eq:Delta_phi_n_i} and \eqref{eq:Delta_A} that $\Delta \Phi_{n^{\!*}}^{i}$ can be rewritten as 
	\begin{align}
	\label{eq:Delta_phi_n_i_2}
		\Delta \Phi_{n^{\!*}}^{i} = - \Phi_{{b_i}}  \sum\limits_{m\in\B}  (\mathbf{A}^{-1})_{n^{\!*}m}  \Delta\!'\! A_{m}^{i},
	\end{align}
	where
	\begin{align}
	\label{eq:Delta_A_2}
		\Delta\!'\! A_{m}^{i}=
		\begin{cases}
			- \teta{i}, & \mathrm{if\ } m=b_i, \\
			- \frac{\hup{m}{i}}{\hup{b_i}{i}} \teta{i}, & \mathrm{if\ } m\neq b_i.
		\end{cases}
	\end{align}	
	Thus, the removal candidate $i^{*}$  among the users having the admission priority of $q$ is obtained as follows:
	\begin{align}
	\label{eq:i_star}
			i^{*} &=\argmax_{i\in\M_q^{\K}}{ |\Delta \Phi_{n^{\!*}}^{i}| } \nonumber \\
			  &=  \argmax_{i\in\M_q^{\K}}{ |\Phi_{{b_i}} \! \sum\limits_{m\in\B} { \!\!\! (\mathbf{A}^{-1})_{n^{\!*}m}  \Delta\!'\! A_{m}^{i}}| }.
	\end{align}
	
	Similar to the uplink communication scenario, the removal candidate $\widetilde{i}^{*}$ at each iteration of the removal process in the downlink scenario is obtained as 
	\begin{align}
	\label{eq:i_stardown}
		\widetilde{i}^{*} &=\argmax_{i\in\M_q^{\K}}{ |\Delta \Pdown_{\widetilde{n}^{\!*}}^{i}| } \nonumber \\
	  &=  
	  \argmax_{i\in\M_q^{\K}}{ | \Pdown_{b_i} \sum\limits_{m\in\B} { \!\!\! (\widetilde{\mathbf{A}}^{-1})_{\widetilde{n}^{\!*}m}  \Delta\!'\! \widetilde{A}_{m}^{i}}| },
	\end{align}
	where $\mathbf{\widetilde{A}=I-\widetilde{H}}(\gammabolddown_{\!\A})$  and $\Pdownbold=\widetilde{\mathbf{A}}^{-1} \widetilde{\mathbf{N}}^{*}$ and
	\begin{align}
	\label{eq:Delta_A_2_down}
		\Delta\!'\! \widetilde{A}_{m}^{i}=
		\begin{cases}
			- \thetadown{i}, & \mathrm{if\ } m=b_i, \\
			- \frac{\hup{i}{b_i}}{\hup{i}{m}} \thetadown{i}, & \mathrm{if\ } m\neq b_i.
		\end{cases}
	\end{align}
	Therefore, MESPA may be revised to obtain the new JPAC algorithm (\textbf{Algorithm 2}).\\
	
	\begin{algorithm}
	
	    \SetAlgoNoLine    
	    \DontPrintSemicolon
	    
	    \small\textbf{Initialization:}\normalsize\;
	    \Indp 
	     	
	     	This is same as the initialization phase of \textbf{MESPA}\;
	     	
	    \Indm
	    	\small \textbf{Admission Control:} \normalsize\;
	    \Indp
	    
	    For the uplink case, calculate $\mathbf{A\!=\!I \! - \! H}(\gammabold_{\!\A})$ from \eqref{eq:phi} and $\Phibold=\mathbf{A}^{-1} \mathbf{N}$ and for the downlink case, calculate $\mathbf{\widetilde{A}\!=\! I \! - \! \widetilde{H}}(\gammabolddown_{\!\A})$ from \eqref{eq:Pdown}  and $\Phibold=\widetilde{\mathbf{A}}^{-1} \widetilde{\mathbf{N}}^{*}$.\;
	    
	    \If {$\gammabold_{\!\A}$ for the uplink case is infeasible (checked through \eqref{eq:fisible_phi}) or $\gammabolddown_{\!\A}$ for the downlink case is infeasible (checked through \eqref{eq:fisible_simple_down})}
	    {
	    	\While{ $|M_q^{\K}\cap \A|=0$}
	    	{
	    		\eIf{$q>1$}
	    		{
	    			$q\leftarrow q-1$\;
	    		}
	    		{
	    			None of the users can be supported,
	    			terminate the algorithm.
	    		}
	    	}
	
	   		Determine $n^{*}\in\B$ for the uplink case from \eqref{eq:ncandid} or $\widetilde{n}^{*}\in\B$ for the downlink case from \eqref{eq:ncandid_down}.\; 
	   		
			Let \mbox{$\A\leftarrow\A \setminus \{i^{*}\}$} for the uplink case, where $i^{*}$ is obtained by \eqref{eq:i_star}, or \mbox{$\A\leftarrow\A \setminus \{\widetilde{i}^{*}\}$} for the downlink case where $\widetilde{i}^{*}$ is obtained by \eqref{eq:i_stardown}\;
	   		
	    }
	    \Indm
	    \small \textbf{Power Calculation:} \normalsize\;
	    \Indp
	 
	 	This is same as the power calculation phase of \textbf{MESPA}.\;
	\caption{MLSPA}
	\end{algorithm}

\subsection{Complexity Analysis of the Proposed Algorithms}

	We consider the worst-case complexity of the algorithms where all users are removed from the set of active users. First note that in the calculation of $\Phibold$ and $\Pdownbold$, the coefficients of $\mathbf{I-H(\gammabold_{\!\A})}$ and $\mathbf{I-\widetilde{H}(\gammabolddown_{\!\A})}$ may be computed once at the initialization phase with complexity of $O(M\!\times\! B)$ and updated for the next steps/iterations with complexity of $O(1)$. In MESPA,  for feasibility checking, at each iteration of the removal process, we need matrix inversion $|\A|$ times. Since there could be at-most $M$ iterations, the overall complexity of MESPA is of $O(M\!\times\! B)+\sum_{|\A|=1}^{M}{|\A|O(B^3)}=O(M^2\!\times\! B^3)$. In MLSPA, we need only one matrix inversion at each of the removal iterations. Therefore, the complexity of MLSPA is obtained as $O(M\!\times\! B)+\sum_{|\A|=1}^{M}{O(B^3)}=O(M\!\times\! B^3)$.
	
\subsection{Comparison Among the Proposed and Existing Algorithms}

	\def\arraystretch{1.2}
	\setlength{\tabcolsep}{3pt}  

	\begin{table*}[t]
		\caption{Comparison of the proposed algorithms with other related algorithms}
		\centering
		\scriptsize

		\begin{tabular}{|c|c|c|c|c|c|c|c|}
		\hline
		\rowcolor[HTML]{EFEFEF} 
		\multicolumn{2}{|c|}{\cellcolor[HTML]{EFEFEF} \textbf{Algorithm} } &
		\textbf{\begin{tabular}[c]{@{}c@{}}
		Complexity of feas.\\ checking of \\ low-priority users
		\end{tabular}}
		 \cellcolor[HTML]{EFEFEF} &
		 \textbf{\begin{tabular}[c]{@{}c@{}}Complexity of feas.\\ checking of \\ high-priority users\end{tabular}}
		\cellcolor[HTML]{EFEFEF} &
		\cellcolor[HTML]{EFEFEF} 
		\textbf{\begin{tabular}[c]{@{}c@{}}
				Complexity of the\\ algorithm
				\end{tabular}}
		& \multicolumn{2}{c|}{\cellcolor[HTML]{EFEFEF}\textbf{\begin{tabular}[c]{@{}c@{}}No. of supported \\ cells   in each \\ prioritized level\end{tabular}}} & \cellcolor[HTML]{EFEFEF} \textbf{\begin{tabular}[c]{@{}c@{}}No. of supported \\ prioritized tiers\end{tabular}}
		\\ \cline{6-7}
		\rowcolor[HTML]{EFEFEF} 
		\multicolumn{2}{|c|}{\multirow{-2}{*}{\cellcolor[HTML]{EFEFEF}}} & \multirow{-2}{*}{\cellcolor[HTML]{EFEFEF}\textbf{\begin{tabular}[c]{@{}c@{}}\end{tabular}}} & \multirow{-2}{*}{\cellcolor[HTML]{EFEFEF}} & \multirow{-2}{*}{\cellcolor[HTML]{EFEFEF}} & \textbf{single-cell} & \textbf{multi-cell} & \multirow{-2}{*}{\cellcolor[HTML]{EFEFEF}} \\ \hline
		\multicolumn{2}{|c|}{ISMIRA \cite{ISMIRA}} & \begin{tabular}[c]{@{}c@{}}iterative and not \\ predictable\end{tabular} & $O(M_s B_p)$ & \begin{tabular}[c]{@{}c@{}}$O\big(M_s(M_s^2 + M_s B_p$ + compl. of \\ feas. checking of low-priority users $\big)$ \\ +  compl. of calculating ITLs\end{tabular} &  & \checkmark & 2 \\ \hline
		\multicolumn{2}{|c|}{CIGSA \cite{JPAC_by_color_graph}} & $O(M_s^3)$ & $O(M_s B_p)$ & \begin{tabular}[c]{@{}c@{}}$O(M_s^3(M_s^3 + M_s B_p))$ + \\ compl. of calculating ITLs\end{tabular} &  & \checkmark & 2 \\ \hline
		\multicolumn{2}{|c|}{CPCSA \cite{JPAC_by_color_graph}} & $O(M_s^3)$ & $O(M_s B_p)$ & \begin{tabular}[c]{@{}c@{}}$O(M_s(M_s^3 + M_s B_p))$ + \\ compl. of calculating ITLs\end{tabular} &  & \checkmark & 2 \\ \hline
		\multicolumn{2}{|c|}{LGRA \cite{LGRA}} & $O(M_s^3)$ & $O(M_s B_p)$ & \begin{tabular}[c]{@{}c@{}}higher than $O(M_s^2(M_s+B_p))$\\ + compl. of calculating ITLs\end{tabular} &  & \checkmark & 2 \\ \hline
		\multicolumn{2}{|c|}{ESRPA \cite{monemi_ESRPA}} & \multicolumn{2}{c|}{integrated into the complexity of the algorithm} & $O(M_s^2)$ & \checkmark &  & 2 \\ \hline
		\multicolumn{2}{|c|}{ELGRA \cite{monemi_ESRPA}} & \multicolumn{2}{c|}{integrated into the complexity of the algorithm} & $O(M_s \log(M_s))$ & \checkmark &  & 2 \\ \hline
		 & \ MESPA \ & \multicolumn{2}{c|}{$O(B^3+ M B)$} & $O(M^2 B^3)$ &  & \checkmark & $K\geq 1$ \\ \cline{2-8} 
		\multirow{-2}{*}{\begin{tabular}[c]{@{}c@{}}Our\\ algorithms\end{tabular}} & \ MLSPA \ & \multicolumn{2}{c|}{$O(B^3+ M B)$} & $O(M B^3)$ &  & \checkmark & $K\geq 1$ \\ \hline
		\end{tabular}

		\label{tbl:comparison_of_algorithms}
	\end{table*}

	Table \ref{tbl:comparison_of_algorithms} shows the main features of our proposed algorithms compared with those of the most well-known JPAC algorithms for two-tier prioritized cellular networks (i.e., cognitive radio networks), namely, ISMIRA \cite{ISMIRA}, CIGSA, and CPCSA \cite{JPAC_by_color_graph}, LGRA \cite{LGRA}, and ESRPA and ELGRA \cite{monemi_ESRPA}\footnote{ ISMIRA \cite{ISMIRA} and LGRA \cite{LGRA} are two well-known power control algorithms for CRNs against which  most of the existing works are compared. CIGSA and CPCSA \cite{JPAC_by_color_graph} are two power control algorithms whose removal metrics are devised based on the theory of graph coloring and have been shown to offer good performance. ESRPA and ELGRA \cite{monemi_ESRPA}  algorithms are computationally very efficient; however, they consider only one primary BS and one cognitive BS.}. In the process of checking whether a subset of admitted low-priority users (i.e., SUs) guarantee the protection of all high-priority users (i.e., PUs), it is assumed in ISMIRA, CIGSA, CPCSA and LGRA that the interference imposed from SUs to each primary BS must be kept under a constant value known as the interference temperature limit (ITL) of that BS. The complexity of checking this  is of $O(M_s)$ for each primary BS and of $O(B_p M_s)$ for all primary BSs at each iteration of the removal process, where $M_s$ and $B_p$ denote the total number of SUs and the total number of primary BSs, respectively. This together with the complexity of the calculation of the ITLs is integrated into the total complexity of the aforementioned algorithms. 

Note that the complexity of ESRPA and ELGRA are much less than that of ISMIRA, CIGSA, CPCSA, and LGRA, but they are designed based on the single-cell model for PRN and CRN and they do not apply to multi-cell primary and secondary networks as other algorithms do. It is seen that the complexities of our proposed algorithms are far below those of others for multi-cell wireless networks. Besides, all algorithms other than ours only work for two-tier cognitive radio networks while both MESPA and MLSPA work for single-tier, two-tier, and prioritized multi-tier networks. 
	
\section{Performance Results and Discussions}
	\label{sec:performance_evaluation_results}

	In this section, we first show the performances of our proposed algorithms for downlink communication scenario in a network consisting of three tiers and then present numerical results to evaluate the performance of our proposed algorithms for the uplink power control scenario in two-tier cognitive radio networks as compared to the well-known existing uplink JPAC algorithms, namely, ISMIRA \cite{ISMIRA} and LGRA \cite{LGRA}.

	For all the following simulation scenarios, the noise power level at the receivers of all links is assumed to be $5\times{10}^{-13}$ Watts and the carrier frequency is 1.9 GHz. The path-gain of the links are considered to be obtained from the statistical path-loss
	and fading model \cite{path_loss_model1} where we set the path-loss exponent to 3 and consider different values for the standard deviation of  log-normal shadowing  for different simulation scenarios. Note that in the simulations we use normalized values of the SINRs. The actual SINR of each user is its normalized SINR multiplied by the processing gain. For example, if the processing gain is equal to 100, the actual SINR is 20 dB higher than the normalized SINR.
	
\begin{figure}
		\centering
		\includegraphics [width=3.5in]{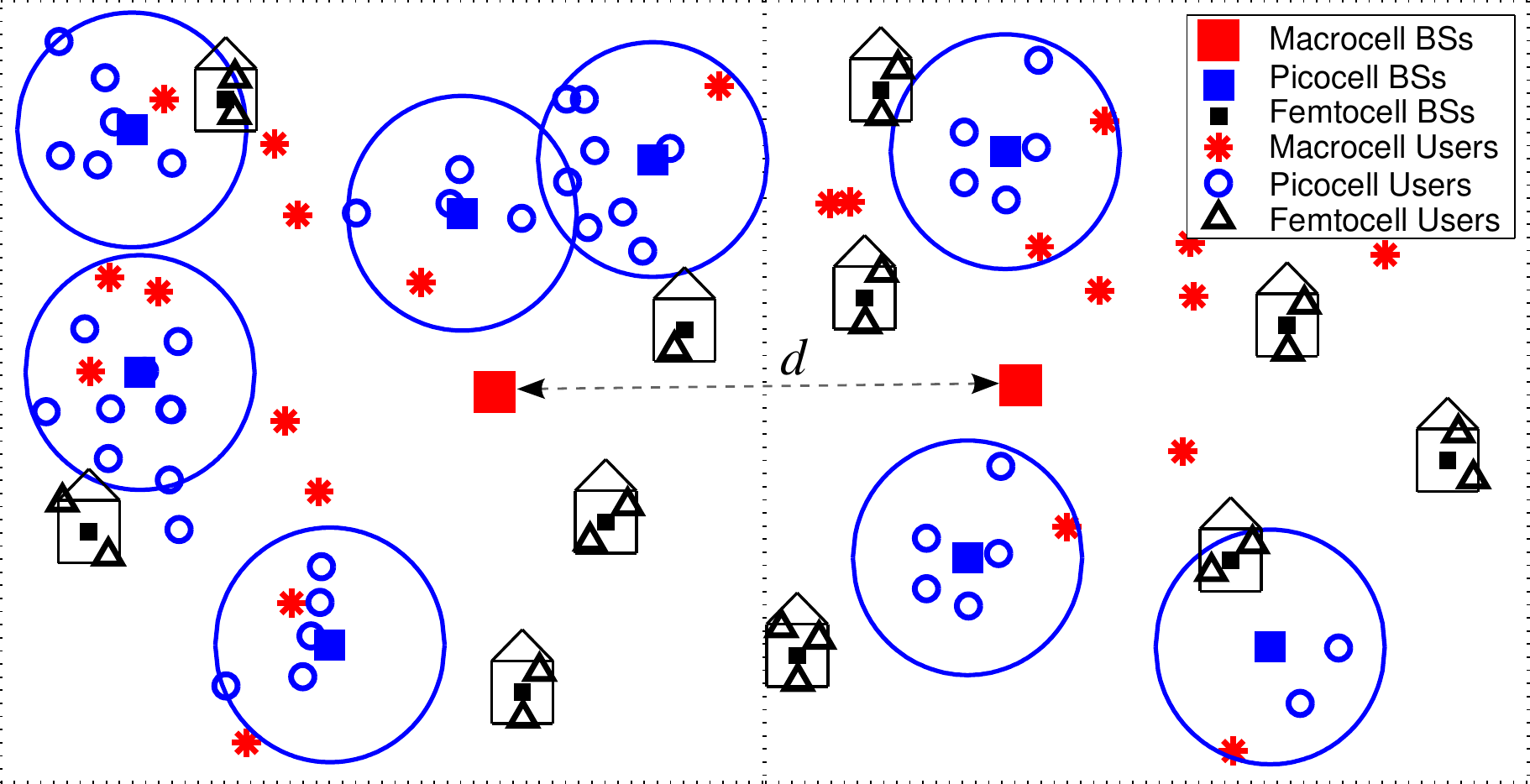}\\ %
		\caption{The three-tier network structure to show the performance of our proposed algorithms in the downlink communication scenario, consisting two macrocells (tier 1), several picocells (tier 2), and several femtocells (tier 3).}
	\label{fig:downlink_sim_structure}
	\end{figure}

\subsection{Downlink Scenario for a Three-tier Network}

	To show the performance of our proposed algorithms in multi-tier downlink communication scenario, we consider a three-tier wireless network as shown in Fig. \ref{fig:downlink_sim_structure} in which there exist two rectangular 1000m $\times$ 1000m macrocells (tier 1), together with several circular picocells (tier 2) each having a radius of 100m and several square-like femtocells (tier 3) each having the dimensions of 20m $\times$ 20m located inside the coverage area of the two macrocells. The BSs of the two macrocells are located at a distance of $d$ m according to the figure. The BSs of the picocells and femtocells together with all users of the macrocells, picocells and femtocells are randomly located inside the corresponding coverage areas according to homogeneous \emph{Poisson Point Processes} (PPPs) with different intensities. The BS of each femtocell and picocell is located at the centre of the corresponding cell. We have considered $d=300$m, the heights of the macrocell, picocell, and femtocell BSs are 20m, 
20m, and 0m, respectively, and the maximum allowed transmit power of the macrocell, picocell, and femtocell BSs are considered to be 50W, 0.5W, and 0.1W, respectively. We also assume that the admission of the macrocell, picocell, and femtocell users has the highest, medium, and lowest priority, respectively. To consider the fading effect of outdoor-to-indoor and indoor-to-outdoor propagations, the standard deviation of log-normal fading is considered to be 6 dB for all links whose one side is within a femtocell and the other side is in some picocell or macrocell. The standard deviation of fading for all other links is considered to be 4 dB.

\begin{figure}
		\centering
		\includegraphics [width=3.5in]{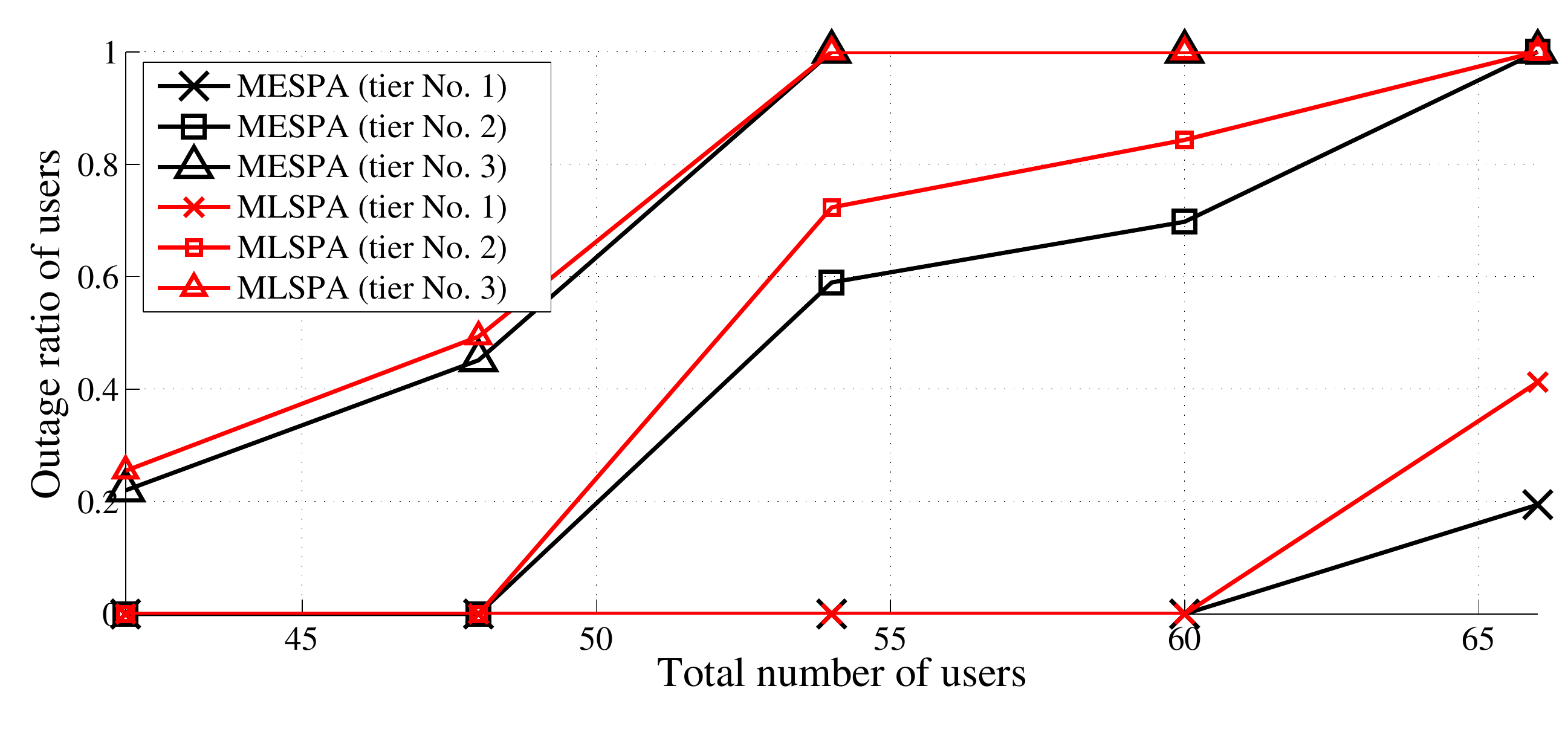}\\ %
		\caption{The outage ratio for each tier versus total number of users of the network for a single snapshot of the downlink three-tier communication for the network according to Fig. \ref{fig:downlink_sim_structure}.}
	\label{fig:downlink_performance_single_snapshot}
	\end{figure}

	Figs. \ref{fig:downlink_performance_single_snapshot} and \ref{fig:downlink_performance_multiple_snapshots} show the performances of our proposed algorithms in terms of  outage ratio of users for a single snapshot and average outage ratio of users for 750 independent snapshots, respectively, versus average total number of users. The outage ratio for each tier at each snapshot is calculated as the ratio of the number of users in that tier who  have been admitted and provided with their target-SINRs and the total number of users of that tier. To obtain the results, we assume that there exist an average of 10, 2, and 2  users initially distributed (according to PPP) in each cell of the tiers 1, 2, and 3, respectively (initially an average of 10, 12, and 20 users in tiers 1, 2, and 3, respectively). Then, at each step, an average of 2 users are added to each of the tiers (i.e., an average of 6 users are added to the network) according to PPP inside some of the cells in the corresponding tier.   It is seen from Fig. \ref{fig:downlink_performance_single_snapshot} that the priority constraints are satisfied by considering that in both MESPA and MLSPA and for each of the total number of users, if the outage of the users of some tier is higher than zero and less than unity, then the outage ratio of tiers of higher and lower priorities are zero and unity, respectively. It is also seen from Fig. \ref{fig:downlink_performance_multiple_snapshots} that, for both MESPA and MLSPA, the average outage ratio of users associated with a higher priority tier is less than that of users associated with  lower priority tiers. Besides,  as seen in both Figs.  \ref{fig:downlink_performance_single_snapshot} and \ref{fig:downlink_performance_multiple_snapshots}, MESPA slightly outperforms MLSPA in terms of outage ratio.
	
	\begin{figure}
		\centering
		\includegraphics [width=3.5in]{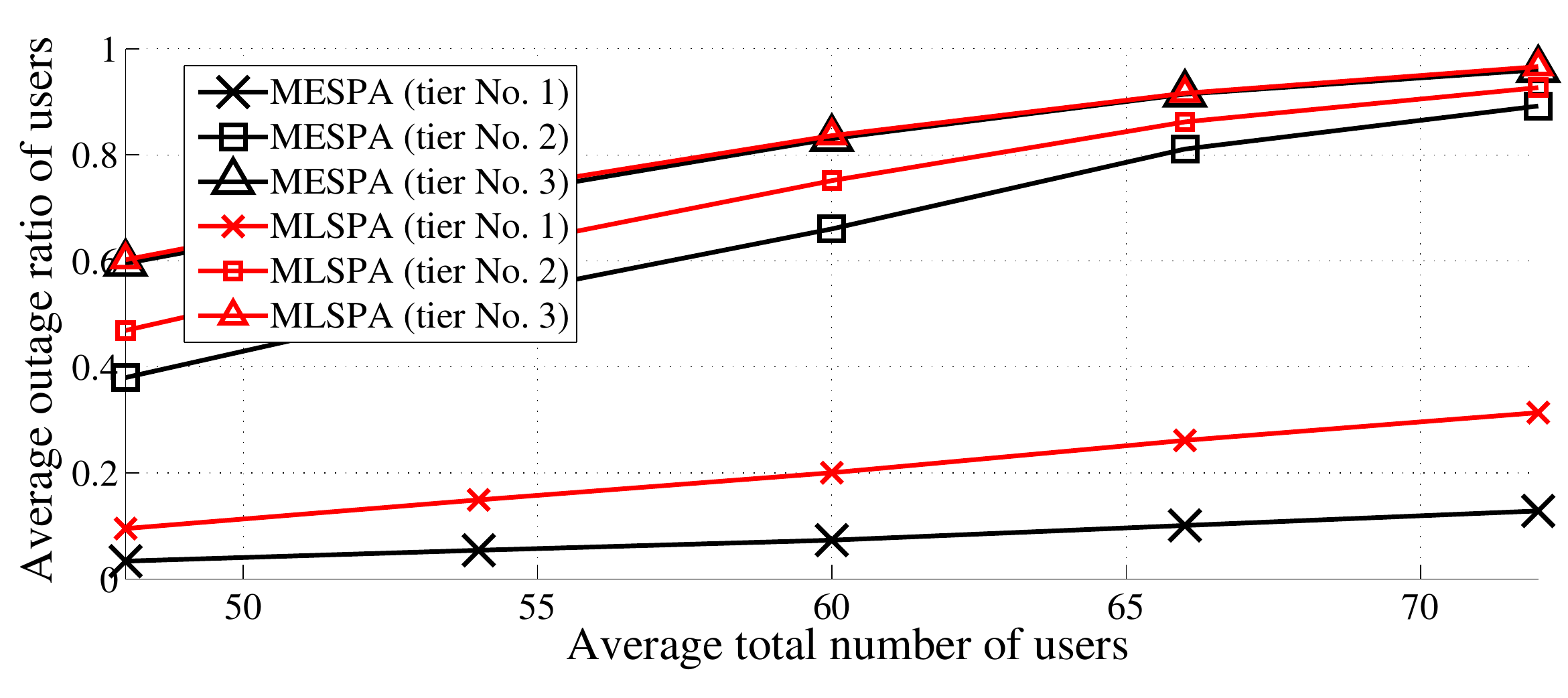}\\ %
		\caption{The average outage ratio for each network tier versus total number of users of the network for 750 independent snapshots of the downlink three-tier communication for the network according to Fig. \ref{fig:downlink_sim_structure}.}
	\label{fig:downlink_performance_multiple_snapshots}
	\end{figure}
	
\subsection{Uplink Scenarios for Two-tier Networks}
	
		\begin{figure*}
			\centering
			\includegraphics []{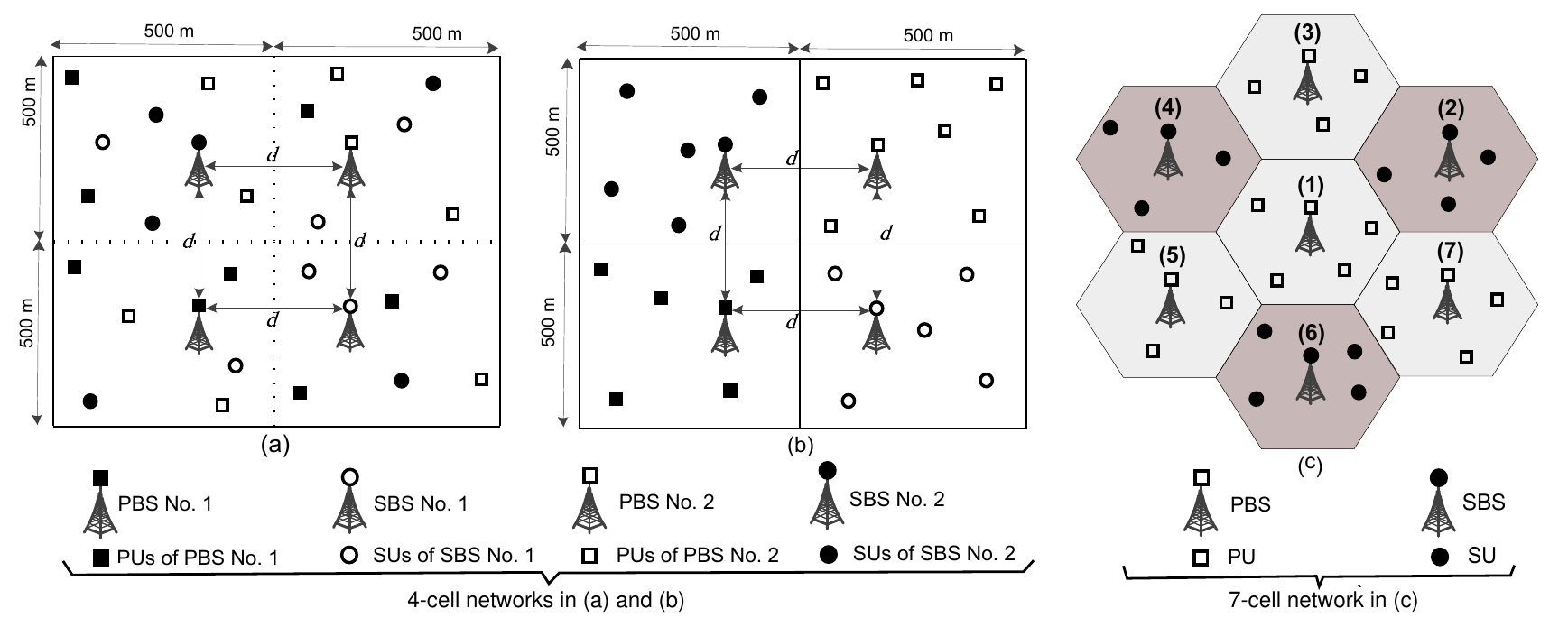}\\
			\caption{Three different scenarios for evaluation of the performances of our proposed algorithms, where  (a) and (b) show 4-cell networks wherein users are randomly spread throughout the network area and in the area closer to their serving BSs, respectively, and (c) shows a network having 7 hexagonal cells consisting of 4 primary and 3 secondary cells.} \vspace{-10pt}
			\label{fig:sim_Topology}
		\end{figure*}

	To compare  MESPA and MLSPA with the existing JPAC algorithms for uplink  two-tier networks (e.g., a network of cognitive radios), we consider three scenarios. In the first two scenarios, we consider two networks each with the area of 1000m $\times$ 1000m. Each of the networks consists of a primary network (PRN) with 2 primary base stations (PBSs) and a CRN with two secondary base stations (SBSs) according to Figs. \ref{fig:sim_Topology}(a) and \ref{fig:sim_Topology}(b). In the third scenario, we consider a 7-cell hexagonal network as shown in Fig. \ref{fig:sim_Topology}(c) in which there exist 4 primary cells (cells 1, 3, 5, and 7) and 3 secondary cells (cells 2, 4, and 6). In Fig. \ref{fig:sim_Topology}(a), all PUs and SUs are randomly spread in the whole coverage area of the network and in Figs. \ref{fig:sim_Topology}(b) and \ref{fig:sim_Topology}(c) all PUs and SUs are randomly located in the area closer to their serving BSs. We have considered $d=150$m in Figs. \ref{fig:sim_Topology}(a) and \ref{fig:sim_Topology}(b), and the radius of each hexagonal cell to be $600$m in Fig. \ref{fig:sim_Topology}(c). The heights of all BSs are considered to be 20 m and the maximum allowed transmit power to be $0.1$W for all users. In all the simulations, we assume that there exist an average of 8 PUs distributed in each cell according to a PPP (an average of 16 PUs in the networks shown in Figs. \ref{fig:sim_Topology}(a) and \ref{fig:sim_Topology}(b), and average of 32  PUs in the network shown in Fig. \ref{fig:sim_Topology}(c)). We evaluate the performances of our proposed algorithms in terms of average outage ratio for the low-priority users (e.g., SUs). The standard deviation of log-normal fading is assumed to be 4 dB for all links. All the results for all of the following simulation scenarios are obtained by averaging over 2500 independent snapshots. 	
	
	\begin{figure}
		\centering
		\includegraphics [width=3.5in]{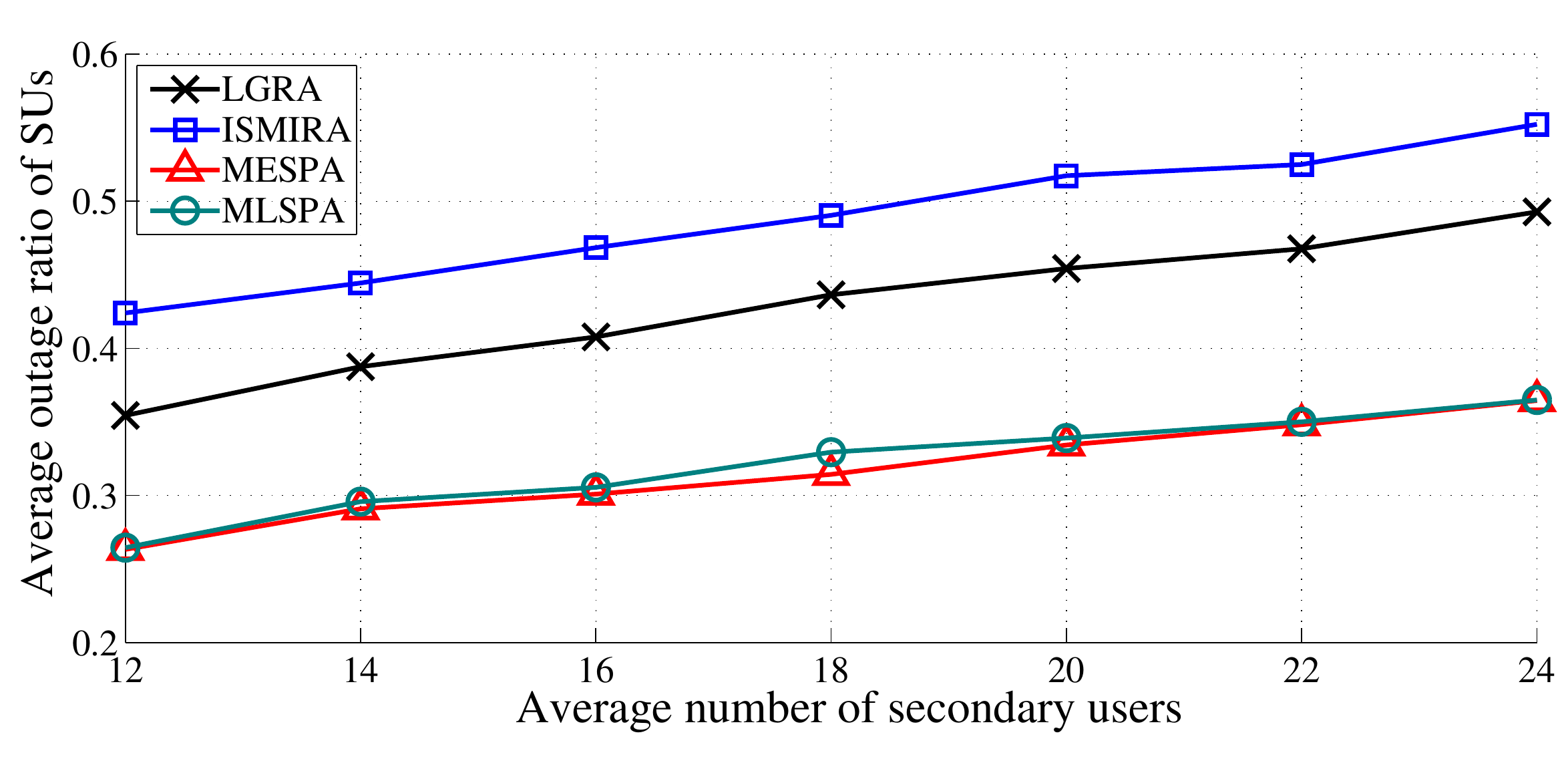}\\ 
		\caption{Average outage ratio for the SUs versus different total number of SUs for the uplink scenario in Fig. \ref{fig:sim_Topology}(a).}
	\label{fig:SUs_outage_4cell1_versus_SUs}
	\end{figure}

\subsubsection{Performance under varying number of low-priority users}
	
		\begin{figure}
		\centering
		\includegraphics [width=3.5in]{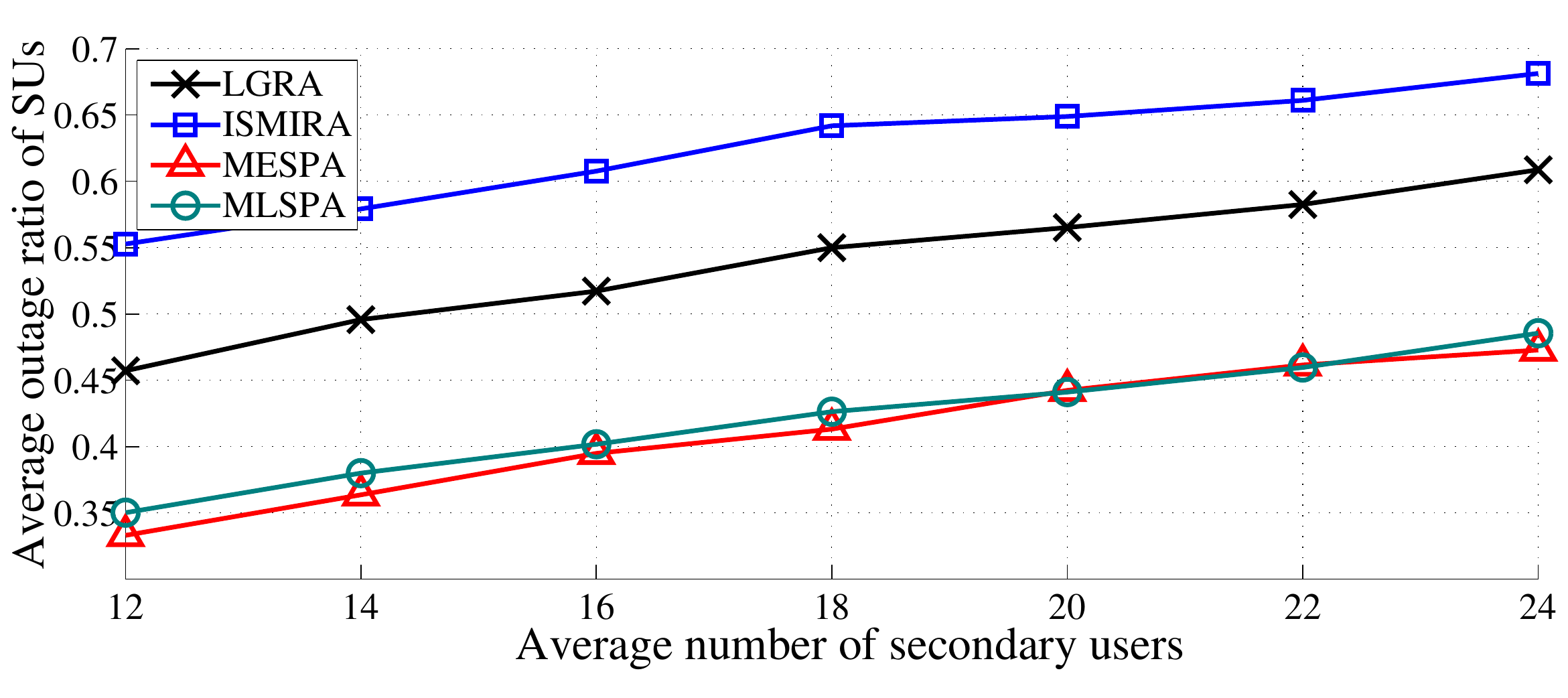}\\ 
		\caption{Average outage ratio for the SUs versus different total number of SUs for the uplink scenario in Fig. \ref{fig:sim_Topology}(b).}
	\label{fig:SUs_outage_4cell2_versus_SUs}
	\end{figure}
	\begin{figure}
		\centering
		\includegraphics [width=3.5in]{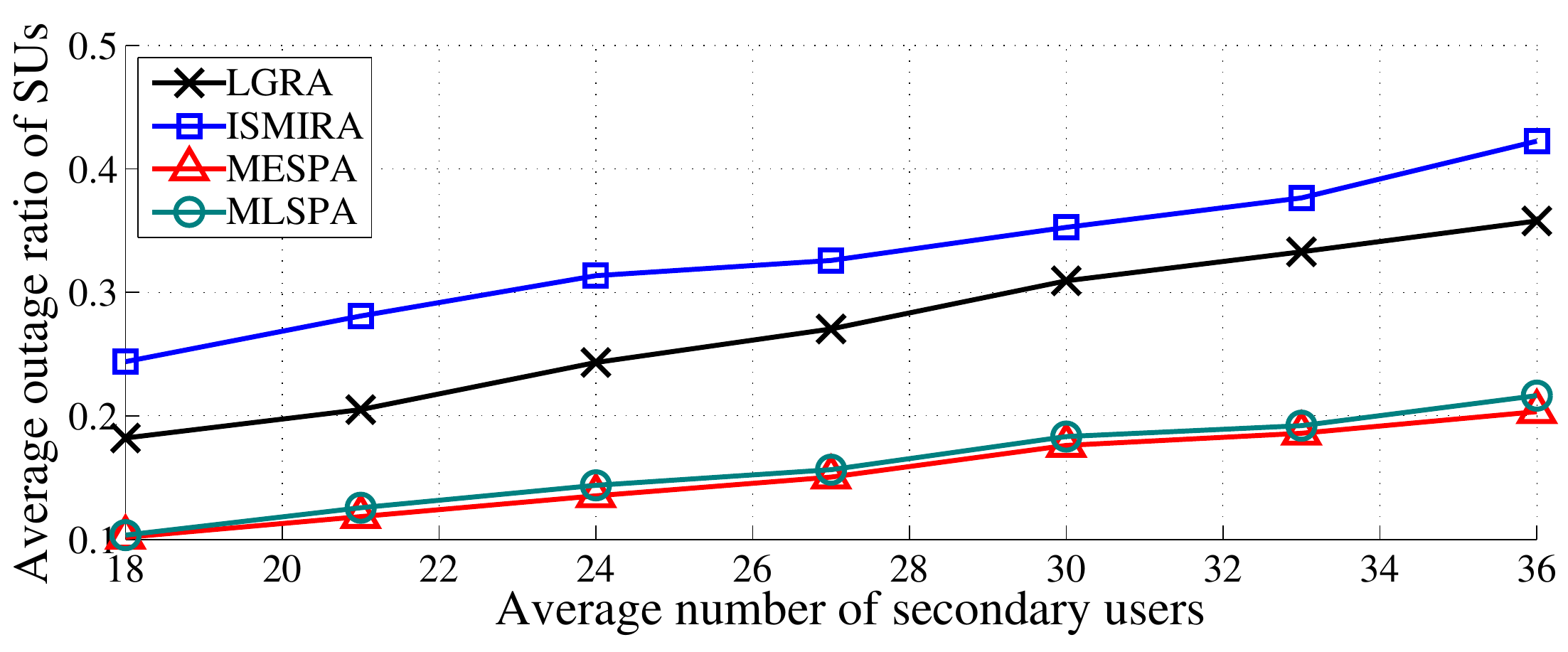}\\ 
		\caption{Average outage ratio for the SUs versus different total number of SUs for the uplink scenario in Fig. \ref{fig:sim_Topology}(c).}
	\label{fig:SUs_outage_Hex_versus_SUs}
	\end{figure}
	
	Consider the case where different number of SUs are distributed according to PPP in each of the cells of the CRN.  The average number of SUs in each cell varies from 6 to 12 with the step-size of one SU (average total of 12 to 24  SUs with the step size of 2 SUs in the 4-cell network shown in Figs. \ref{fig:sim_Topology}(a) and \ref{fig:sim_Topology}(b), and average total of 18 to 36 SUs with the step size of 3 SUs in the  network shown in Fig. \ref{fig:sim_Topology}(c)). We consider that the target-SINR for each user is randomly chosen from the set $\{-16, -22\}$ dB for the 4-cell network shown in Fig. \ref{fig:sim_Topology}(a), and from the set $\{-10, -16\}$ dB for the networks shown in Figs. \ref{fig:sim_Topology}(b) and \ref{fig:sim_Topology}(c). Figs. \ref{fig:SUs_outage_4cell1_versus_SUs}, \ref{fig:SUs_outage_4cell2_versus_SUs}, and \ref{fig:SUs_outage_Hex_versus_SUs} show the average outage ratio of SUs for the networks shown in Figs. \ref{fig:sim_Topology}(a), \ref{fig:sim_Topology}(b), and \ref{fig:sim_Topology}(c), respectively. In addition to the lower computational complexity of MESPA and MLSPA as shown in Table \ref{tbl:comparison_of_algorithms}, from  all of these figures it is observed that, with these algorithms, a higher average number of SUs can be supported when compared with ISMIRA and LGRA,  for all values of total number of SUs in all the scenarios. It is also seen that MESPA offers a slightly superior performance when compared with MLSPA at the cost of a higher computational complexity. 

\subsubsection{Performance under varying target-SINR of users}

	\begin{figure}
		\centering
		\includegraphics [width=3.5in]{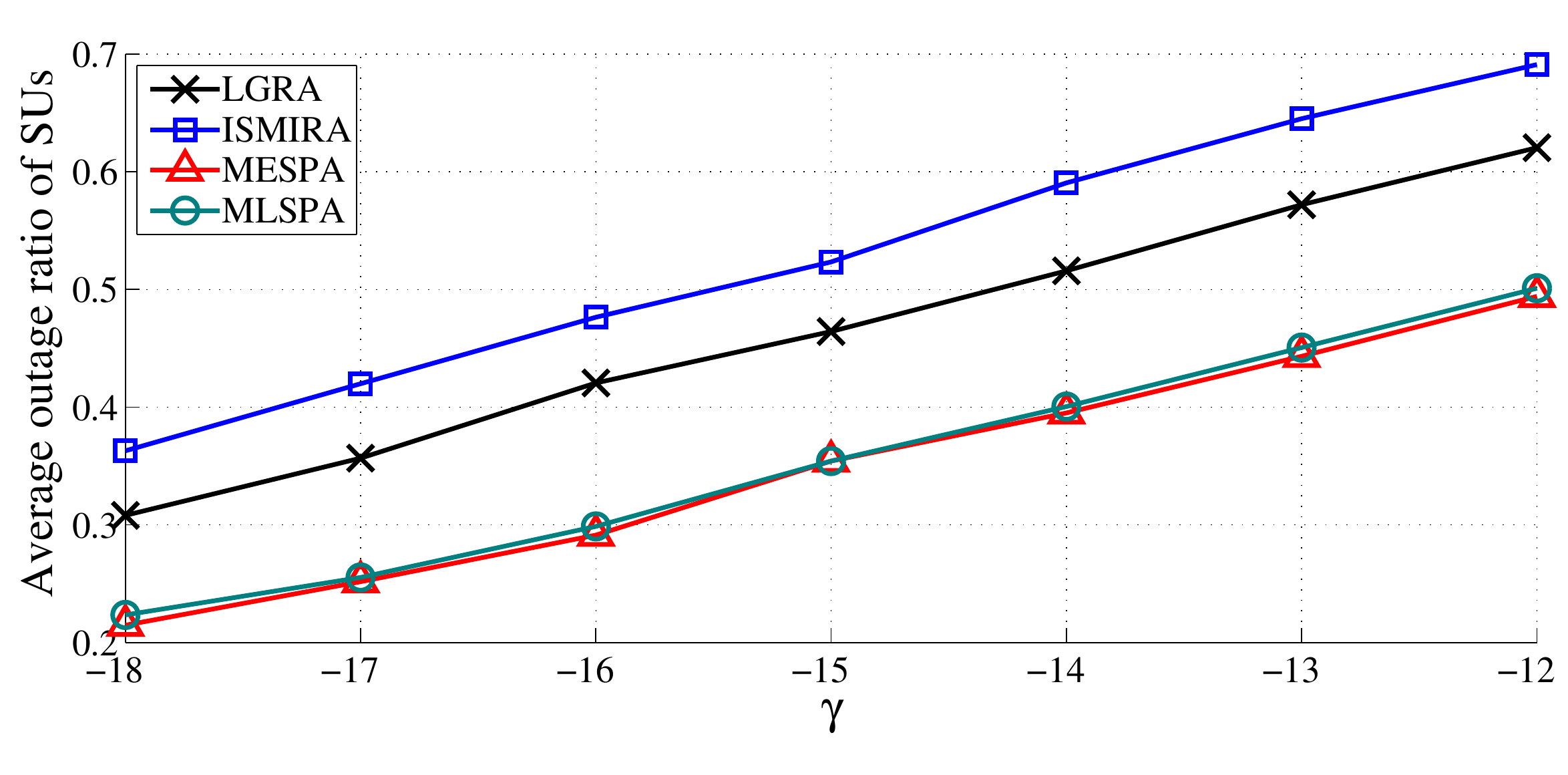}\\ 
		\caption{Average outage ratio for the SUs versus different values of $\gamma$ for the uplink scenario in Fig. \ref{fig:sim_Topology}(a).}
	\label{fig:SUs_outage_4cell1_versus_gamma}
	\end{figure}
	\begin{figure}
		\centering
		\includegraphics [width=3.5in]{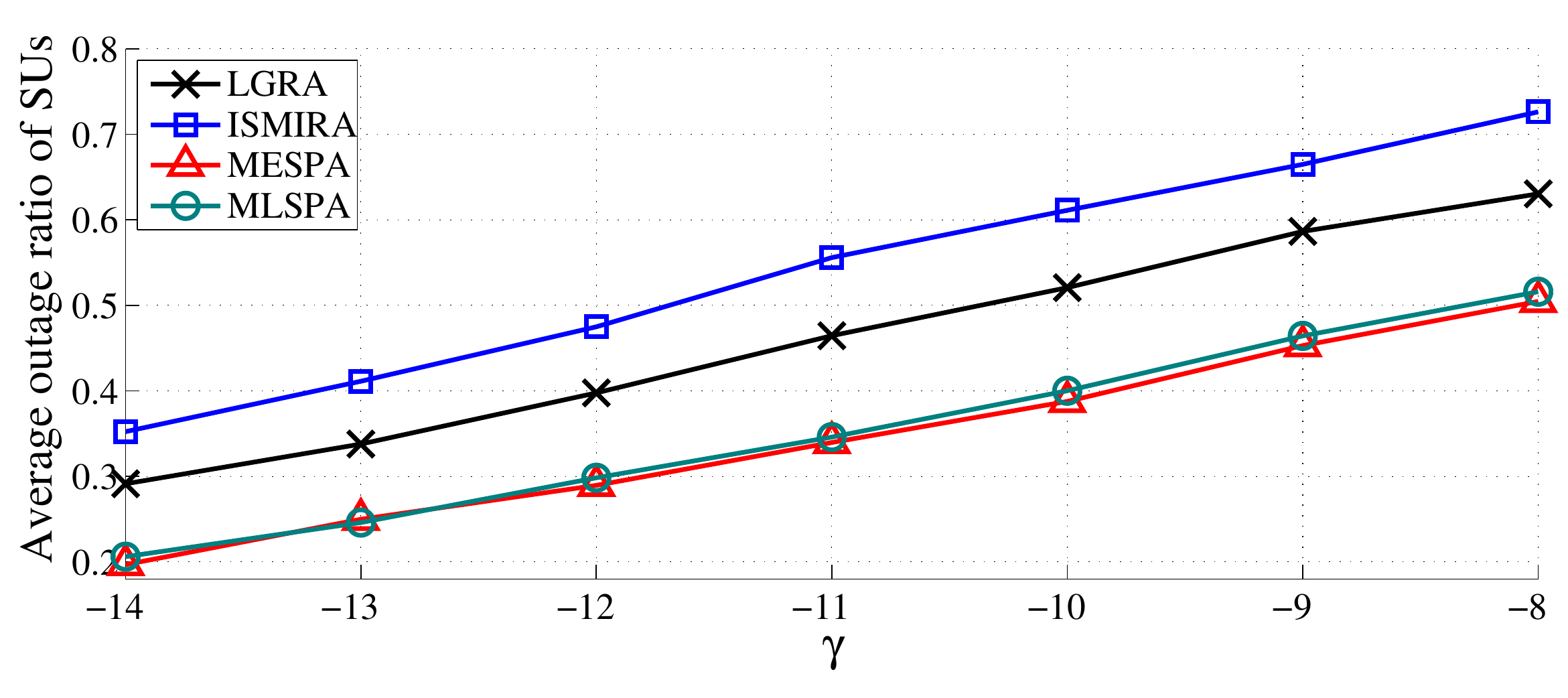}\\ 
		\caption{Average outage ratio for the SUs versus different values of $\gamma$ for the uplink scenario in Fig. \ref{fig:sim_Topology}(b).}
	\label{fig:SUs_outage_4cell2_versus_gamma}
	\end{figure}
	\begin{figure}
		\centering
		\includegraphics [width=3.5in]{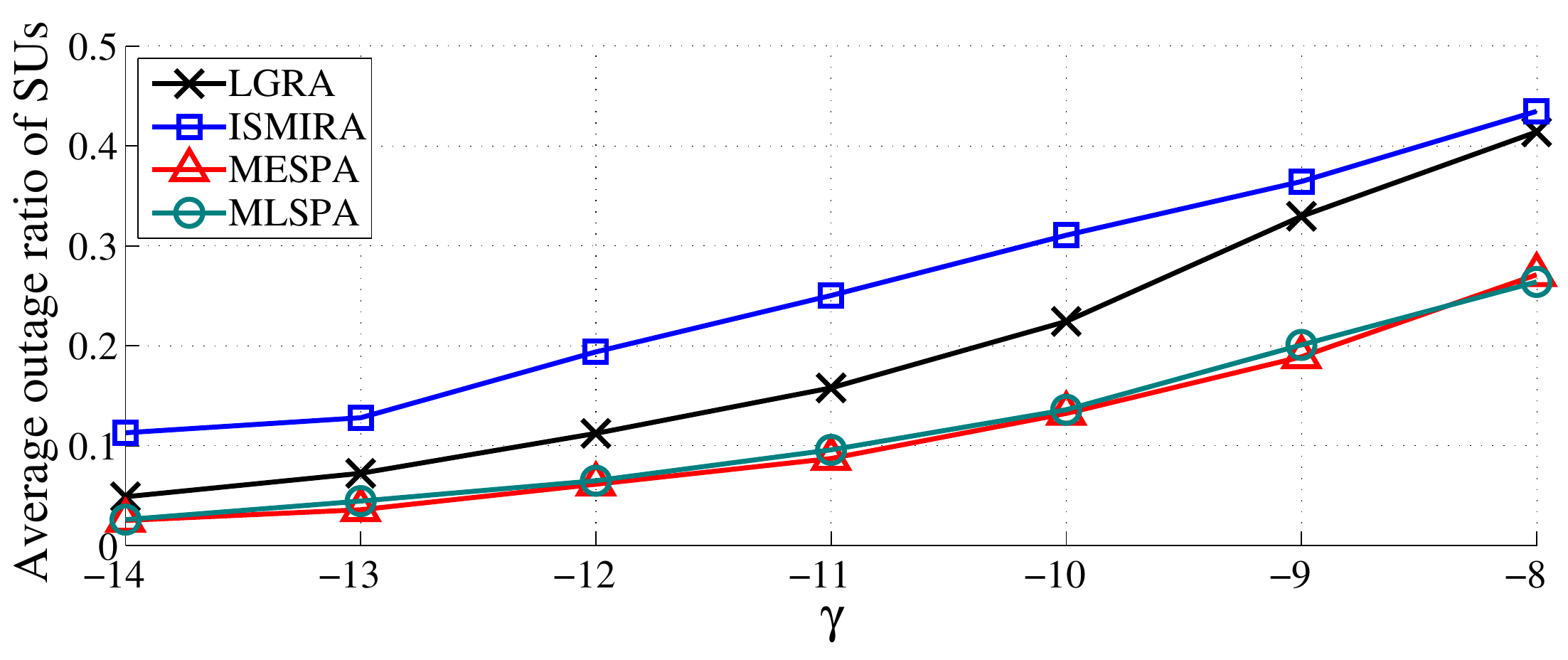}\\ 
		\caption{Average outage ratio for the SUs versus different values of $\gamma$ for the uplink scenario in Fig. \ref{fig:sim_Topology}(c).}
	\label{fig:SUs_outage_Hex_versus_gamma}
	\end{figure}
	
	To show how the performances of the algorithms are affected by different values of the target-SINRs of the users, we consider the case where there exist an average of 8 SUs distributed (according to PPP) in each cell  (average of 16  SUs in the networks according to Figs. \ref{fig:sim_Topology}(a) and \ref{fig:sim_Topology}(b), and average of 24  SUs in the network according to Fig. \ref{fig:sim_Topology}(c)). We consider that the target-SINR of each user (PU or SU) is randomly chosen from the set of $\{\gamma, \gamma-6\}$ dB in all scenarios for all algorithms. Here $\gamma$ varies from -18 to -12 dB with the step size of 1 dB for the network in Fig. \ref{fig:sim_Topology}(a), and varies from -14 to -8 dB with the step size of 1 dB for the networks in  Figs. \ref{fig:sim_Topology}(b) and \ref{fig:sim_Topology}(c). Figs. \ref{fig:SUs_outage_4cell1_versus_gamma}, \ref{fig:SUs_outage_4cell2_versus_gamma}, and \ref{fig:SUs_outage_Hex_versus_gamma} show the average outage ratio of SUs of the networks in Figs. \ref{fig:sim_Topology}(a), \ref{fig:sim_Topology}(b), and \ref{fig:sim_Topology}(c), respectively. It is observed  that our proposed algorithms offer a lower  average outage ratio for the SUs in comparison to that for ISMIRA and LGRA for all values of $\gamma$.
	
\subsubsection{Performance under varying standard deviation of shadow fading}

	\begin{figure}
		\centering
		\includegraphics [width=3.5in]{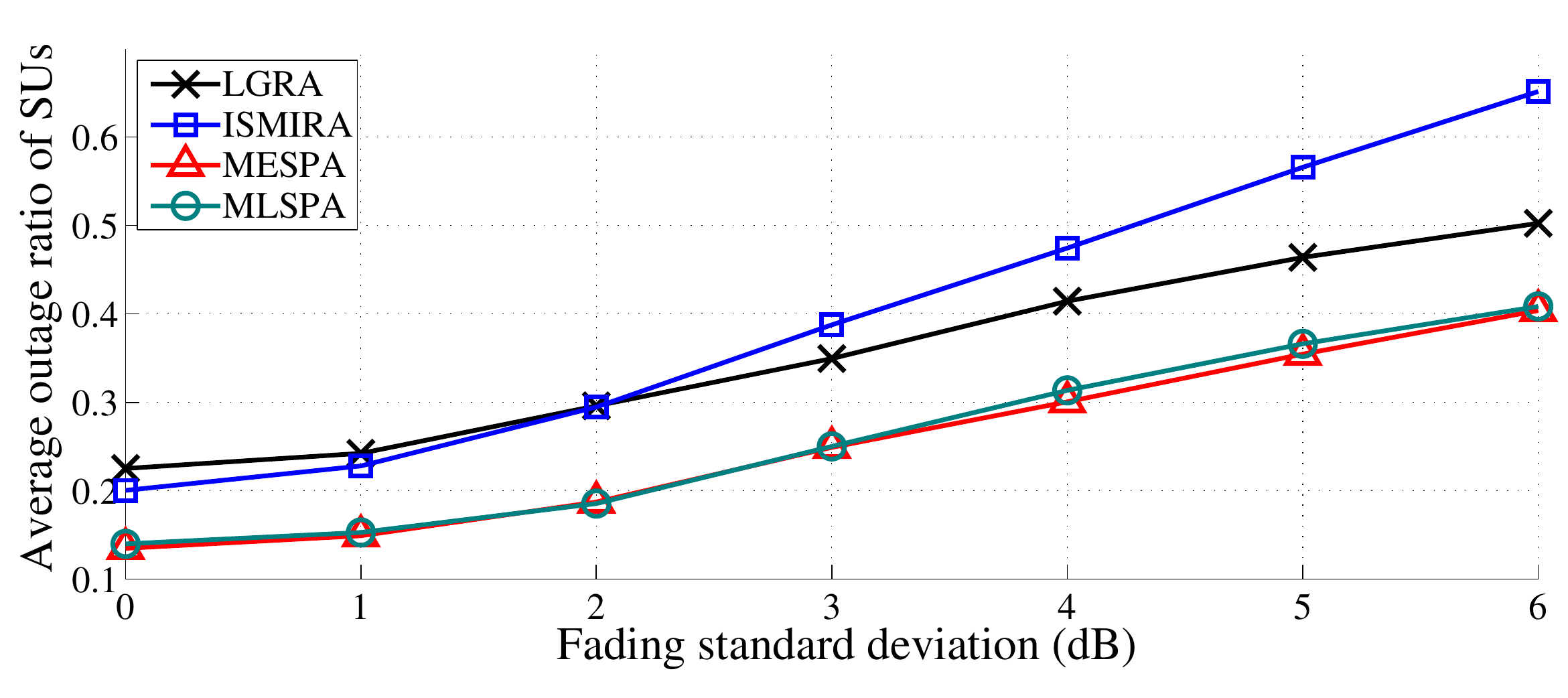}\\ 
		\caption{Average outage ratio for the SUs versus different values of the standard deviation of the log-normal fading for the uplink scenario in Fig. \ref{fig:sim_Topology}(a).}
	\label{fig:SUs_outage_4cell1_versus_fading}
	\end{figure}
	\begin{figure}
		\centering
		\includegraphics [width=3.5in]{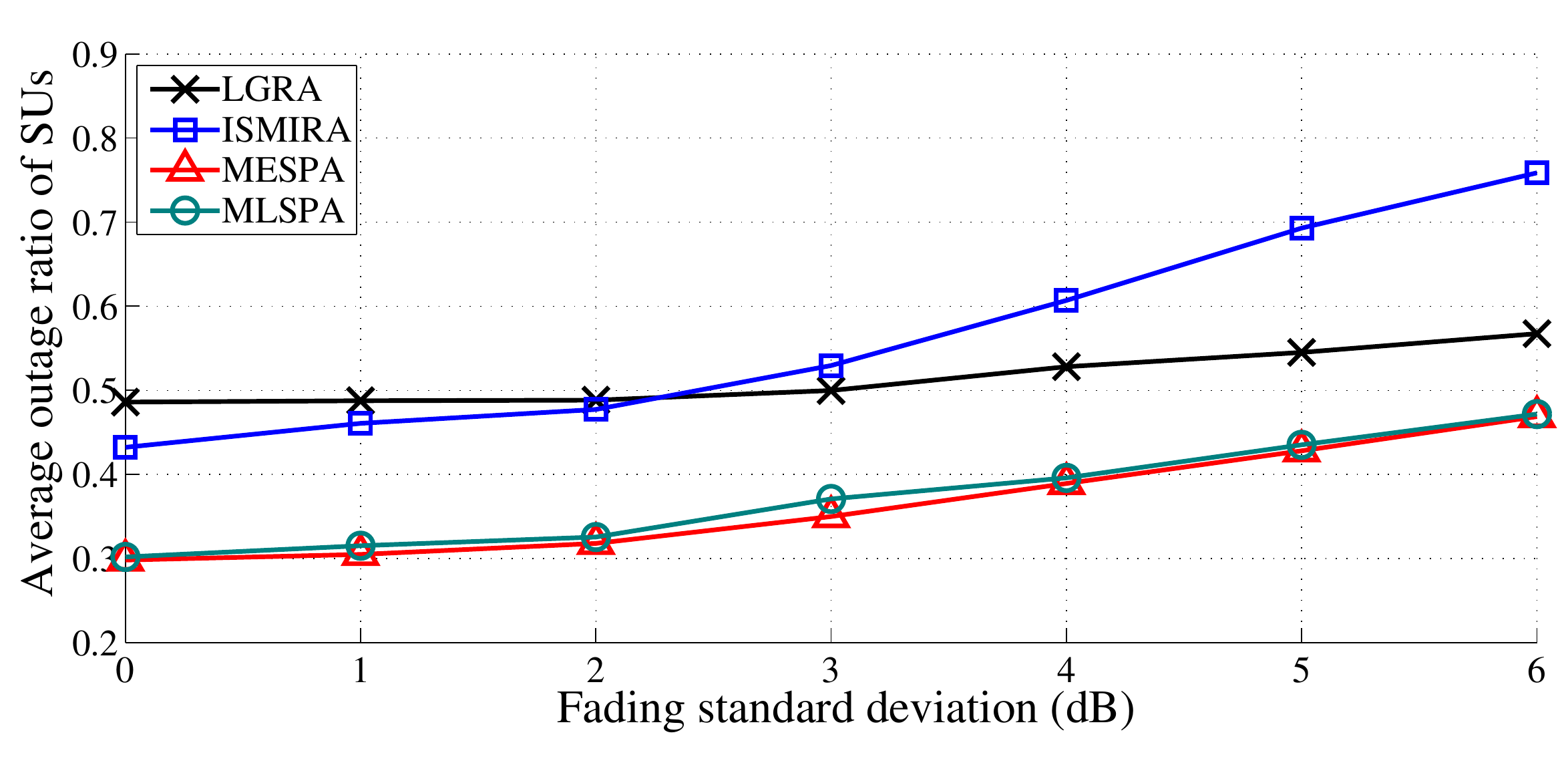}\\ 
		\caption{Average outage ratio for the SUs versus different values of the standard deviation of the log-normal fading for the uplink scenario in Fig. \ref{fig:sim_Topology}(b).}
	\label{fig:SUs_outage_4cell2_versus_fading}
	\end{figure}
	\begin{figure}
		\centering
		\includegraphics [width=3.5in]{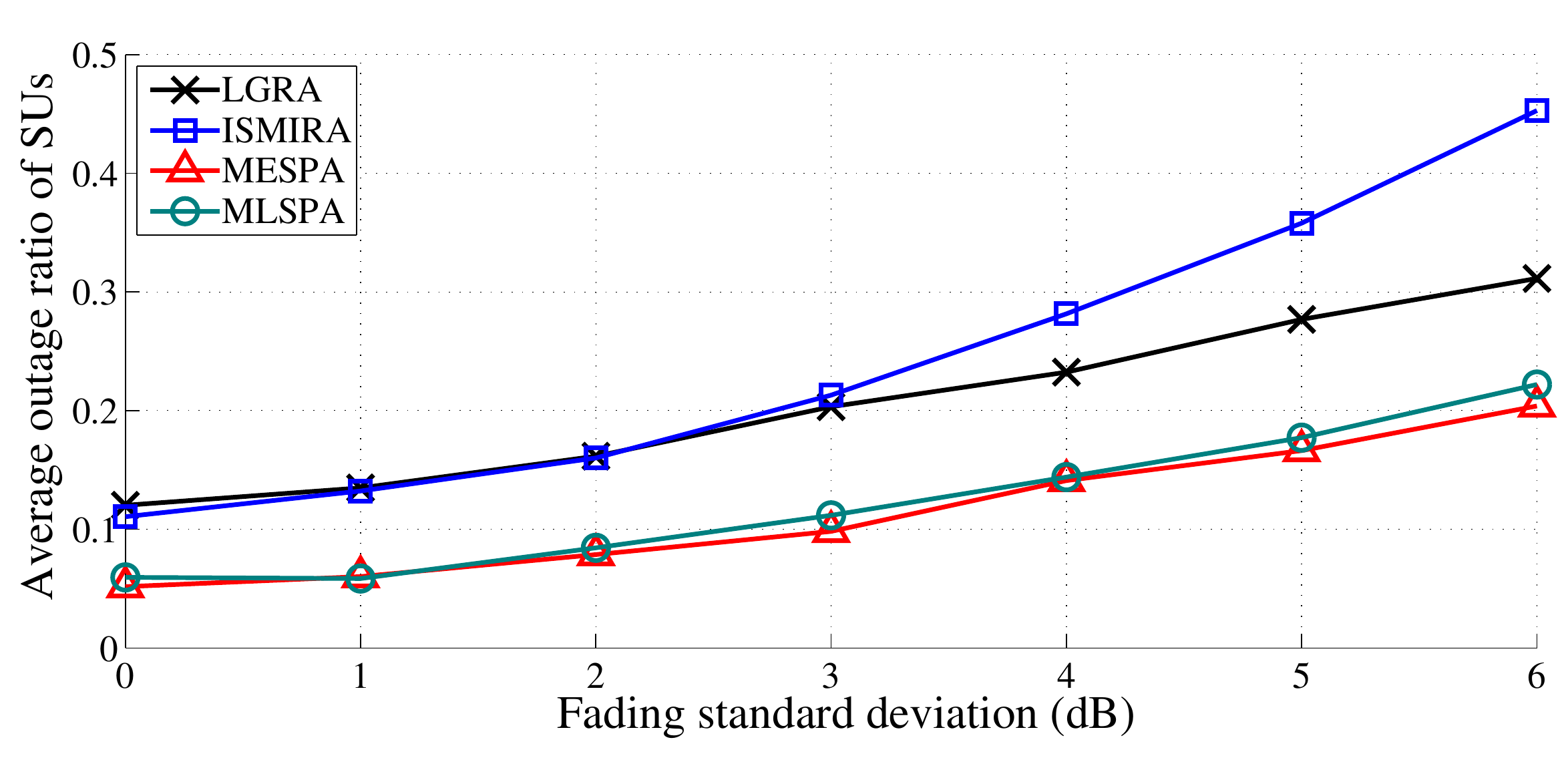}\\ 
		\caption{Average outage ratio for the SUs versus different values of the standard deviation of the log-normal fading for the uplink scenario in Fig. \ref{fig:sim_Topology}(c).}
	\label{fig:SUs_outage_Hex_versus_fading}
	\end{figure}
	
	Finally, the effect of different values of the standard deviation of log-normal fading is studied. As in the previous scenario, we consider the case where there exist an average of 8 SUs distributed according to PPP in each cell  (average of 16  SUs in the networks shown in Figs. \ref{fig:sim_Topology}(a) and \ref{fig:sim_Topology}(b), and average of 24  SUs in the network shown in Fig. \ref{fig:sim_Topology}(c)).  We consider that the target-SINR for each user is randomly chosen from the set $\{-16, -22\}$ dB for the 4-cell network shown in Fig. \ref{fig:sim_Topology}(a), and from the set $\{-10, -16\}$ dB for the networks shown in Figs. \ref{fig:sim_Topology}(b) and \ref{fig:sim_Topology}(c). Figs. \ref{fig:SUs_outage_4cell1_versus_fading}, \ref{fig:SUs_outage_4cell2_versus_fading}, and \ref{fig:SUs_outage_Hex_versus_fading} show the average outage ratio of the SUs for the networks scenarios according to Figs. \ref{fig:sim_Topology}(a), \ref{fig:sim_Topology}(b), and \ref{fig:sim_Topology}(c), respectively, versus different values of the standard deviation of log-normal fading varying from 0 to 6 dB with the step size of 1 dB. It is seen that as the standard deviation of the fading increases, all algorithms show higher average outage ratio. Besides, the performances of our proposed algorithms  are always better than those of ISMIRA and LGRA.

\section{Conclusion}	
\label{sec:conclusions}
We have studied the problem of obtaining the maximum feasible set of users in prioritized multi-tier infrastructure-based cellular networks. We have first obtained a simple relationship between a given SINR vector and its corresponding uplink/downlink power vector based on which we have devised two novel JPAC algorithms for underlay cellular wireless networks. To the best of our knowledge, the proposed algorithms are the first JPAC algorithms for prioritized multi-tier cellular networks in the literature supporting more than two priority levels. The complexities of our proposed algorithms are much lower than those of existing algorithms. Numerical results show that, when compared to existing algorithms, our proposed algorithms support higher number of users with their QoS requirements by considering the priority constraints. The effects of channel gain uncertainties on the performances of the proposed JPAC algorithms and development of  robust  JPAC algorithms for prioritized multi-tier cellular networks will be considered in our future work.
	
\begin{appendices}
	
	\section{Proof of Proposition \ref{prop:2}}
		\label{apx:1}
		Let $\Phi_{b_i}$ be the total received power plus noise at the BS serving user $i$, i.e.,
		\begin{align}
		\label{eq:phi_i}
			\Phi_{b_i}=\sum\limits_{j\in\M}{\!\!\hup{b_i}{j} p_j} \!  + \Nup{b_i}.
		\end{align}
		From {\eqref{eq:1}} and {\eqref{eq:phi_i}}, we have 
		\begin{equation}
		\label{eq:11}
			\gammai=
			\dfrac	{\hup{b_i}{i} p_i}
					{\Phi_{b_i}-\hup{b_i}{i} p_i}, \ \forall i\in\M.
		\end{equation}
		This results in
		\begin{align}
		\label{eq:pi_versus_phi}
			p_i= \teta{i} \dfrac{\Phi_{b_i}}{\hup{b_i}{i}}. 
		\end{align}
		From {\eqref{eq:pi_versus_phi}}, for each $m,n\in\B$, the following is obtained:
		\begin{align}
		\label{eq:22}
			\sum\limits_{i\in\M_n^{\B}}{p_i \hup{m}{i}} = \Phi_{n} \sum\limits_{i\in\M_n^{\B}}{\frac{\hup{m}{i}}{\hup{n}{i}}\teta{i}}.
		\end{align}
		By letting $m=n$ and adding $\sum_{i\notin\M_m^{\B}}{\!\!p_i \hup{m}{i}}+N_{m}$ to both sides of {\eqref{eq:22}}, $\Phi_{m}$ is obtained as
		\begin{align}
		\label{eq:23}
			\Phi_{m}
			= & \frac{\sum\limits_{i\notin\M_m^{\B}} {\!\!\!\big(p_i \hup{m}{i} \big)} + N_{m}} {1- \! \sum\limits_{i\in\M_m^{\B}}\!\!\! \left(\teta{i}\right)}.
		\end{align}
		From {\eqref{eq:22}} and {\eqref{eq:23}} we have
		\begin{align}
		\label{eq:24}
			\Phi_{m}
			= & \dfrac{\sum\limits_{\substack{n\in\B\\ n\neq m}} \ \sum\limits_{i\in\M_n^{\B}} {\!\!\!\big(p_i \hup{m}{i} \big)} + N_{m}} {1- \! \sum\limits_{i\in\M_m^{\B}}\!\!\! \left(\teta{i}\right)} \nonumber \\
			= & 
			\frac{\sum\limits_{\substack{n\in\B\\ n\neq m}} \ \Phi_{n} \sum\limits_{i\in\M_n^{\B}}{ \frac{\hup{m}{i}}{\hup{n}{i}}\teta{i}} + N_{m}} {1- \! \sum\limits_{i\in\M_m^{\B}}\!\!\! \left(\teta{i}\right)}.
		\end{align}
		This results in
		\begin{multline}
			\label{eq:25}
			\Phi_{m} \left( 1- \! \sum\limits_{i\in\M_m^{\B}}\!\!\! \left(\teta{i}\right) \right) -
			\sum\limits_{\substack{n\in\B\\ n\neq m}} \ \Phi_{n} \sum\limits_{i\in\M_n^{\B}}{ \dfrac{\hup{m}{i}}{\hup{n}{i}}\teta{i}}
					\\	= N_{m},\ \ \ \  m=1,2,...,B.
		\end{multline}
		Writing {\eqref{eq:25}} in matrix form results in {\eqref{eq:phi}}.
	\section {Proof of Proposition \ref{prop:3}}
	\label{apx:2}
		Let $\Phidown_{i}$ be the total received power plus noise at user $i$. We have
		\begin{align}
		\label{eq:phi_i_down}
			\Phidown_{i} &=\sum\limits_{j\in\M}{\hdown{i}{b_j} p_j} \!  + \Ndown{i}.
		\end{align}
		Similar to \eqref{eq:pi_versus_phi}, from \eqref{eq:1down} we have
		\begin{align}
		\label{eq:pidown_versus_phi}
			\pdown_i= \dfrac{\gammadown_i}{(\gammadown_i+1)} \dfrac{\Phidown_{i}}{\hdown{i}{b_i}}, \ \ \ \forall i\in\M.
		\end{align}	
		\eqref{eq:phi_i_down} can be rewritten as
		\begin{align}
			\label{eq:phi_i_down_sum_P}
				\Phidown_{i} &=\sum\limits_{j\in\M}{\hdown{i}{b_j} p_j} \!  + \Ndown{i} \nonumber \\
					&= \sum\limits_{j\in\M_1}{\!\!\! \hdown{i}{b_j} p_j} + 
		 	\sum\limits_{j\in\M_2}{\!\!\!\hdown{i}{b_j} p_j} + \dots + 
		 	\sum\limits_{j\in\M_B}{\!\!\!\hdown{i}{b_j} p_j} + \Ndown{i} \nonumber \\
				&= \hdown{i}{1}\!\!\sum\limits_{j\in\M_1}{\!\! p_j} + 
		 	\hdown{i}{2}\!\!\sum\limits_{j\in\M_2}{\!\! p_j} + \dots +  
		 	\hdown{i}{B}\!\!\sum\limits_{j\in\M_B}{\!\! p_j} + \Ndown{i} \nonumber \\
				&= 	\sum\limits_{n\in\B} \hdown{i}{n}\Pdown_n + \Ndown{i}.
		\end{align}
		From \eqref{eq:pidown_versus_phi} and \eqref{eq:phi_i_down_sum_P}, for each $i\in\M$ we have
		\begin{align}
		\label{eq:pidown_versus_Pdown}
				\pdown_i= \frac{1}{\hdown{i}{b_i}} \thetadown{i} \times \left( \sum\limits_{n\in\B}{ \hdown{i}{n}\Pdown_n } + \Ndown{i} \right).
		\end{align}	
		For any $m\in\B$, summing the two parts of \eqref{eq:pidown_versus_Pdown} over all $i\in\M_m^{\B}$ results in
		\begin{align}
		\label{eq:pidown_versus_Pdown_sum}
			\sum\limits_{i\in\M_m^{\B}} \! \pdown_i   
			&= \!  \sum\limits_{i\in\M_m^{\B}}
				{ \sum\limits_{n\in\B}  \dfrac{\hdown{i}{n}}{\hdown{i}{b_i}} \thetadown{i} \Pdown_n  
	 } + \!\!
				\sum_{i\in\M_m^{\B}} \! \! \frac{1}{\hdown{i}{b_i}} \thetadown{i} \Ndown{i} 
				\nonumber \\
			&= \!  \sum\limits_{n\in\B}
						{ \Pdown_n  \left( \sum\limits_{i\in\M_m^{\B}}  \dfrac{\hdown{i}{n}}{\hdown{i}{b_i}} \thetadown{i}   \right)
	 } \!\! + \!\!
				\sum_{i\in\M_m^{\B}} \!\! \frac{1}{\hdown{i}{b_i}} \thetadown{i} \Ndown{i}.
		\end{align}
		Therefore,
		\begin{multline}
		\label{eq:pidown_versus_Pdown_sum2}
			\Pdown_m  =    \sum\limits_{\substack{n\in\B \\ n\neq m}}
				{ \Pdown_n \!\! \left( \sum\limits_{i\in\M_m^{\B}} \dfrac{\hdown{i}{n}}{\hdown{i}{m}} \thetadown{i}  \right)}  
				+\Pdown_m \!\sum\limits_{i\in\M_m^{\B}} \!\thetadown{i}
				\\
 				+ 
				\!\sum_{i\in\M_m^{\B}} \!\!\!\ \frac{1}{\hdown{i}{m}} \thetadown{i}\Ndown{i},
		\end{multline}
		and thus
		\begin{multline}
		\label{eq:pidown_versus_Pdown_sum3}
			\Pdown_m  \left( 1 \!- \!\!\!\sum\limits_{i\in\M_m^{\B}}\thetadown{i} \!\!\right) 
			-   \sum\limits_{\substack{n\in\B \\ n\neq m}}
				{  \Pdown_n  \left( \sum\limits_{i\in\M_m^{\B}}  \dfrac{\hdown{i}{n}}{\hdown{i}{m}} \thetadown{i}  \right)} 
				= \\
				\sum_{i\in\M_m^{\B}} \!\! \frac{1}{\hdown{i}{m}} \thetadown{i}\Ndown{i}.
		\end{multline}
		Writing \eqref{eq:pidown_versus_Pdown_sum3} in matrix form results in \eqref{eq:Pdown}.
\end{appendices}

\bibliographystyle{IEEEtran}
\bibliography{Mybib}

\begin{biography}[{\includegraphics[width=1in,height=1.25in,clip,keepaspectratio]{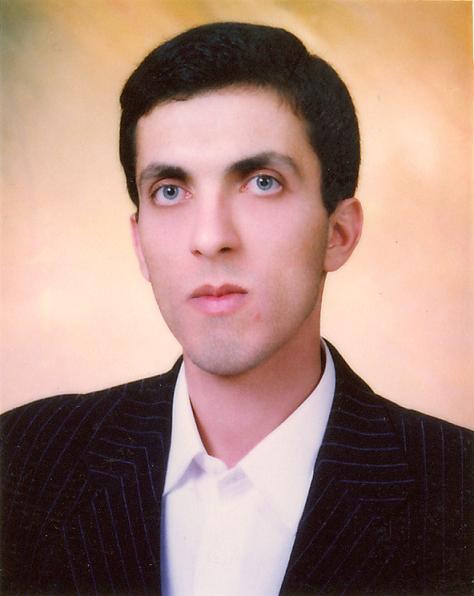}}]{Mehdi Monemi}
    received the B.Sc., M.Sc. degrees all in electrical and computer engineering from Shiraz University, Shiraz, Iran, and Tarbiat Modares University, Tehran, Iran, in 2001 and 2003, respectively. He has
    recently received Ph.D. degree in electrical and computer engineering from Shiraz University, Shiraz, Iran. His current research interests include resource allocation in wireless networks, and traffic engineering in computer networks.
\end{biography}

\begin{biography}[{\includegraphics[width=1in,height=1.25in,clip,keepaspectratio]{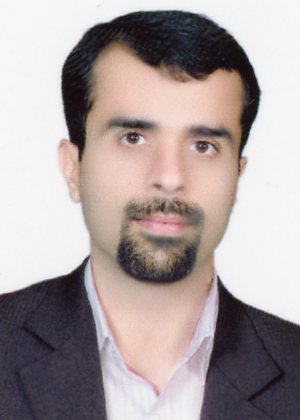}}]{Mehdi Rasti} (S'08-M'11) received his B.Sc. degree from Shiraz University, Shiraz, Iran, and the M.Sc. and Ph.D. degrees both from Tarbiat Modares University, Tehran, Iran, all in Electrical Engineering in 2001, 2003 and 2009, respectively. 
From November 2007 to November 2008, he was a visiting researcher at the Wireless@KTH, Royal Institute of Technology, Stockholm, Sweden. From September 2010 to July 2012 he was with Shiraz University of Technology, Shiraz, Iran, after that he joined the Department of Computer Engineering and Information Technology, Amirkabir University of Technology, Tehran, Iran, where he is now an assistant professor. 
From June 2013 to August 2013, and from July 2014 to August 2014 he was a visiting researcher in the Department of Electrical and Computer Engineering, University of Manitoba, Winnipeg, MB, Canada. His current research interests include 
radio resource allocation in wireless networks and network security.
\end{biography}

\begin{biography} [{\includegraphics[width=1in,height=1.25in,clip,keepaspectratio]{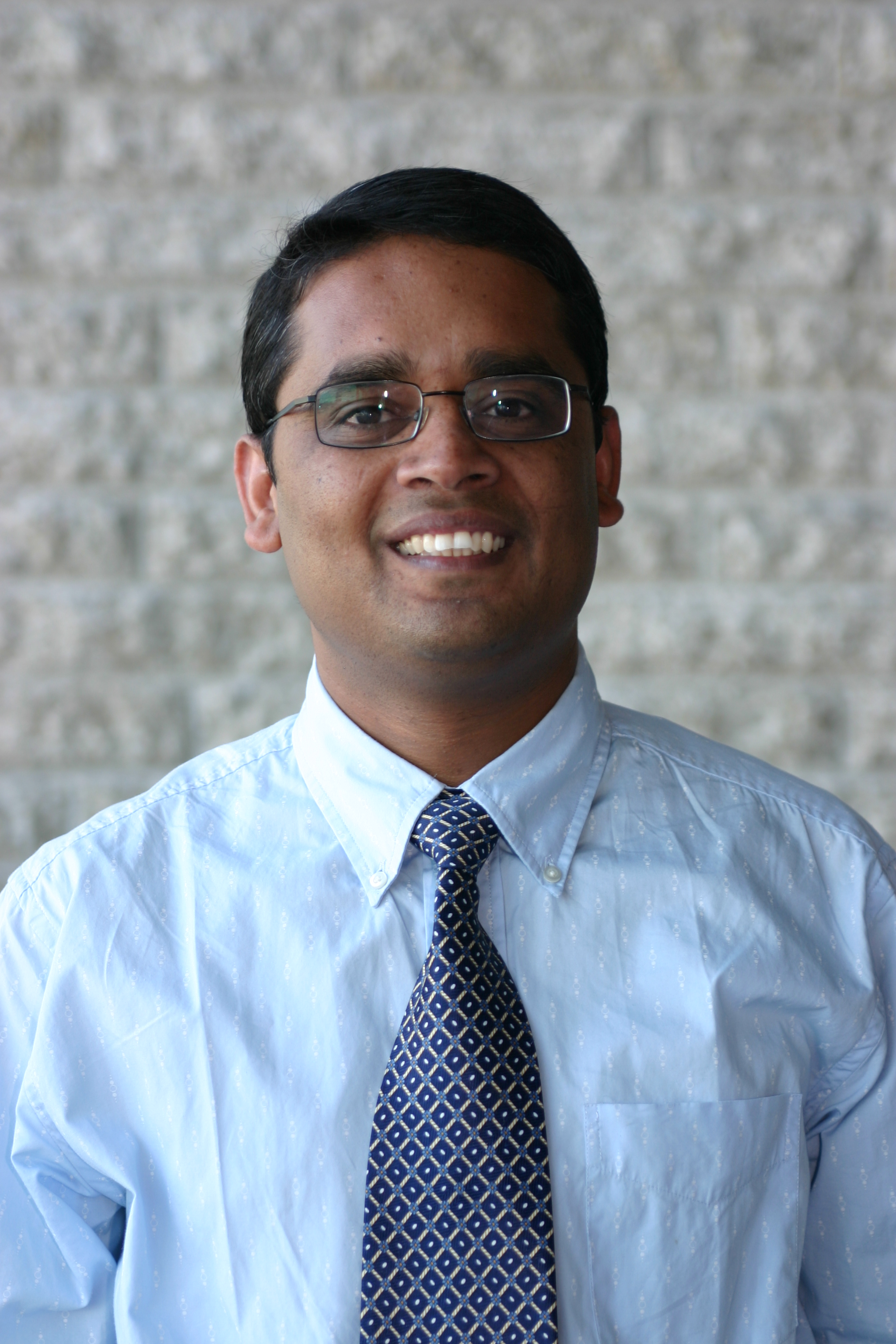}}]
{Ekram Hossain} (F'15)
is a Professor in the Department of Electrical and Computer
Engineering at University of Manitoba, Winnipeg,
Canada. He received his Ph.D. in Electrical
Engineering from University of Victoria,
Canada, in 2001. Dr. Hossain's current research
interests include design, analysis, and optimization
of wireless/mobile communications networks, cognitive
radio systems, and network economics. He
has authored/edited several books in these areas
(http://home.cc.umanitoba.ca/$\sim$hossaina). He was elevated to an IEEE Fellow ``for contributions to spectrum management and resource allocation in cognitive and cellular radio networks". 
Dr. Hossain has won several research awards including
the IEEE Communications Society Transmission, Access, and Optical Systems (TAOS) Technical Committee's Best Paper Award in IEEE Globecom 2015, University of Manitoba Merit Award in 2010 and 2014 (for Research and
Scholarly Activities), the 2011 IEEE Communications Society Fred Ellersick
Prize Paper Award, and the IEEE Wireless Communications and Networking
Conference 2012 (WCNC'12) Best Paper Award. He is a Distinguished Lecturer of the
IEEE Communications Society (2012-2015). He is a registered Professional
Engineer in the province of Manitoba, Canada.
\end{biography}

\end{document}